\documentclass[a4paper,11pt]{article}
\pdfoutput=1
\usepackage[utf8]{inputenc}
\usepackage[T1]{fontenc}
\usepackage[UKenglish]{isodate}

\usepackage{mathrsfs}
\usepackage{bm}
\usepackage{bbm}
\usepackage{palatino}

\usepackage[margin=3cm]{geometry}

\usepackage[usenames,dvipsnames]{xcolor}
\usepackage{amsmath}
\usepackage{amssymb}
\usepackage{amsthm}  
\usepackage{ifthen}
\usepackage{mathabx}
\usepackage{stmaryrd}
\usepackage{array}
\usepackage{float}
\usepackage[colorlinks=true,urlcolor=Mahogany,linkcolor=Mahogany,citecolor=Mahogany,plainpages=false,pdfpagelabels]{hyperref}
\usepackage{tikz}
\usepackage{verbatim}
\usetikzlibrary{shapes.geometric,plotmarks,backgrounds,fit,calc,circuits.ee.IEC}
\usetikzlibrary{decorations.pathreplacing}
\usepackage[backend=bibtex,maxcitenames=4,maxalphanames=100,maxbibnames=100,isbn=false, sorting=none]{biblatex}
\tikzset{phase/.style = {draw,fill,shape=circle,minimum size=5pt,inner sep=0pt},crossx/.style={path picture={ 
\draw[thick,black,inner sep=0pt]
(path picture bounding box.south east) -- (path picture bounding box.north west) (path picture bounding box.south west) -- (path picture bounding box.north east);
}}, cross/.style={path picture={ 
\draw[thick,black](path picture bounding box.north) -- (path picture bounding box.south) (path picture bounding box.west) -- (path picture bounding box.east);
}}, not/.style={draw,circle,cross,minimum width=0.3 cm}}
\usepackage{xcolor}

\usepackage[colorinlistoftodos]{todonotes}

\usepackage{mathtools}

\newtheorem{theorem}{Theorem}
\newtheorem{lemma}[theorem]{Lemma}

\newtheorem{proposition}[theorem]{Proposition}
\newtheorem{definition}[theorem]{Definition}

\newcommand{\sket}[1]{{\ensuremath{\lvert#1\rangle}}}
\newcommand{\lket}[1]{{\ensuremath{\left\lvert#1\right\rangle}}}
\newcommand{\ket}[1]{\mathchoice{\lket{#1}}{\sket{#1}}{\sket{#1}}{\sket{#1}}}

\newcommand{\sbra}[1]{{\ensuremath{\langle#1\rvert}}}
\newcommand{\lbra}[1]{{\ensuremath{\left\langle#1\right\rvert}}}
\newcommand{\bra}[1]{\mathchoice{\lbra{#1}}{\sbra{#1}}{\sbra{#1}}{\sbra{#1}}}

\newcommand{\stackindx}[2]{\text{\renewcommand{\arraystretch}{0.5}$\begin{array}{@{}c@{}} #1 \\ #2 \end{array}$}}

\DeclareFieldFormat{eprint:hal}{%
  HAL Id\addcolon\space
  \ifhyperref
      {\href{https://hal.inria.fr/hal-#1}{\texttt{#1}}}
    {\texttt{#1}}}

\DeclareFieldFormat{eprint:iacr}{%
  IACR eprint\addcolon\space
  \ifhyperref
    {\href{http://eprint.iacr.org/#1}{\texttt{#1}}}
    {\texttt{#1}}}

\DeclareFieldFormat{eprint:hal}{%
  HAL Id\addcolon\space
  \ifhyperref
      {\href{https://hal.archives-ouvertes.fr/#1}{\texttt{#1}}}
    {\texttt{#1}}}

\newcommand{\ident}{\mathbbm{1}}

\newcommand{\email}[1]{\href{mailto:#1}{#1}}
\renewcommand{\otimes}{\varotimes}

\newcommand{\eqdef}{:=}

\begin{document}
\cleanlookdateon
\allowdisplaybreaks

\makeatletter
\def\thanks#1{\protected@xdef\@thanks{\@thanks
        \protect\footnotetext{#1}}}
\makeatother

\title{Multilevel Polarization for Quantum Channels}
\author{Ashutosh~Goswami, Mehdi~Mhalla, Valentin~Savin \vspace*{-10mm}
\thanks{A. Goswami is with  Université Grenoble Alpes, Grenoble INP, LIG, F-38000 Grenoble, France (\email{ashutosh-kumar.goswami@univ-grenoble-alpes.fr}).}
\thanks{M. Mhalla is with Université Grenoble Alpes, CNRS, Grenoble INP, LIG, F-38000 Grenoble, France (\email{mehdi.mhalla@univ-grenoble-alpes.fr}).}
\thanks{V. Savin is with Université Grenoble Alpes, CEA-LETI, F-38054 Grenoble, France (\email{valentin.savin@cea.fr}).}}
\maketitle

\begin{abstract}

Recently, a purely quantum version of polar codes has been proposed in~\cite{DGMS19} based on a quantum channel combining and splitting procedure, where a randomly chosen two-qubit Clifford unitary acts as channel combining operation. Here, we consider the quantum polar code construction using the same channel combining and splitting procedure as in~\cite{DGMS19}, but with a fixed two-qubit Clifford unitary. For the family of Pauli channels, we show that polarization happens in multi-levels, where synthesized quantum virtual channels tend to become completely noisy, half-noisy, or noiseless. Further, we present a quantum polar code exploiting the multilevel nature of polarization, and provide an efficient decoding for this code. We show that half-noisy channels can be frozen by fixing their inputs in either the amplitude or the phase basis, which allows reducing the number of preshared EPR pairs compared to the construction in~\cite{DGMS19}. We provide an upper bound on the number of preshared EPR pairs, which is an equality in the case of the quantum erasure channel.  To improve the speed of polarization, we propose an alternative construction, which again polarizes in multi-levels, and the previous upper bound on the number of preshared EPR pairs also holds. For a quantum erasure channel, we confirm by numerical analysis that the multilevel polarization happens relatively faster for the alternative construction.
%

\end{abstract}

\section{Introduction}

Polar codes are a family of the capacity-achieving codes for any  discrete memoryless classical channel, with efficient encoding and decoding algorithms~\cite{arikan09, sta09}. Polar codes have been generalized for quantum channels in two different ways. The first generalization is a CSS-like construction, which uses the polar codes for classical-quantum (cq) channels in the amplitude and the phase basis \cite{rdr11, rw12, wg13-2, wg13}. The CSS-like construction achieves symmetric coherent information for any qubit-input quantum channel and has an efficient decoding algorithm for the Pauli channel. Recently, a new generalization is proposed in~\cite{DGMS19, our-itw-paper19}, which is called \emph{purely} quantum polar codes. The purely quantum construction relies on a specific quantum channel combining and splitting procedure, where a randomly chosen two-qubit Clifford unitary combines two copies of a quantum channel. The recursive channel combining and splitting procedure synthesizes so called virtual channels, which tend to be either ``completely noisy'', or ``noiseless'' as quantum channels, not merely in one basis, hence, the name purely quantum. The code is entanglement assisted as preshared EPR pairs need to be supplied for all the completely noisy channels. This construction also achieves symmetric coherent information for any qubit-input quantum channel and has an efficient decoding in the case of the Pauli channel \cite{DGMS19}. Moreover, it is shown that choosing the channel combining operation from a set of 9 or 3 two-qubit Clifford unitaries is sufficient to achieve polarization for the Pauli channel.

\medskip In this work, we consider the following two questions arising naturally from~\cite{DGMS19}:
\begin{itemize}
\item Whether polarization still can be achieved when the channel combining operation is a fixed two-qubit Clifford unitary.

\item How much we can reduce the number of preshared EPR pairs. 
\end{itemize}
For the first question, we show that the Pauli channel polarizes, using the channel combining and splitting procedure defined in~\cite{DGMS19}, but with a fixed two-qubit Clifford gate as channel combining operation. However, polarization here happens in multi-levels in the sense of~\cite{PB12, sahebi2011multilevel}, instead of two levels. In particular, the synthesized virtual channels can also be ``half-noisy''  except being completely noisy or noiseless. The half-noisy channels need to be frozen by fixing their inputs in either the amplitude or the phase basis, while preshared EPR pairs are required for the completely noisy channels as in~\cite{DGMS19}. As some of the bad channels are frozen in either the amplitude or the phase basis, the quantum polar code constructed here requires a fewer number of preshared EPR pairs than the construction in~\cite{DGMS19}. We also give an upper bound on the number of preshared EPR pairs, which is an equality for the quantum erasure channel. In particular, for a quantum erasure channel with erasure probability $\epsilon$, the fraction of preshared EPR pairs is $\epsilon^2$, while it is $\epsilon$ for the construction proposed in~\cite{DGMS19}. Therefore, for the second question, the number of preshared EPR pairs is significantly reduced, taking advantage of the multilevel nature of polarization. The decoding can also be efficiently performed by decoding a classical polar code on a classical channel with  a 4-symbol input alphabet similar to~\cite{DGMS19}. 

Finally, we relax the fixed channel combining condition and present a slightly different construction utilizing a quantum circuit equivalence. For a quantum erasure channel, we show with the help of a computer program that the multilevel polarization occurs relatively faster for this alternative construction compared to the first construction. Further, the alternative construction requires the same number of preshared EPR pairs as the first construction.

The paper is organized as follows: in Section~\ref{sec:preliminary}, we recall some useful definitions, and properties of the quantum polar code proposed in~\cite{DGMS19}. The definitions of the symmetric mutual information and the Bhattacharyya parameter of a classical channel are also provided. In Section~\ref{sec:noisy_half_noiseless}, we introduce noiseless, half-noisy, and noisy channels. In Section~\ref{sec:multilevel}, we prove our main result, that is, the multilevel polarization in the case of the Pauli channel, using a fixed two-qubit Clifford as channel combining operation. In Section~\ref{sec:Qpolar_constr}, it is shown that the multilevel polarization can be used to construct an efficient quantum polar code. We also give an upper bound on the number of preshared EPR pairs and a fast polarization property that ensures reliable decoding. Finally, in Section~\ref{sec:alt_constr}, we propose an alternative construction to improve the speed of polarization, and in Section~\ref{sec:quantm_erasure}, it is shown by numerical simulation that for a quantum erasure channel, the speed of polarization significantly improves, when the alternative construction is used instead of the first construction.

\section{Preliminaries} \label{sec:preliminary}

\noindent \textbf{Notation:} Let $P_N$ be the $N$-qubit Pauli group, $\mathcal{C}_N$ be the $N$-qubit Clifford group, and $\bar{P}_N = P_N/\{\pm 1,  \pm i\}$ be the Abelian group obtained by taking the quotient of $P_N$ by its centralizer. We write $\bar{P}_1 = \{I, X, Y, Z \}$, and $\bar{P}_2 = \{ u \otimes v | u, v \in P_1 \} \cong \bar{P}_1 \times \bar{P}_1$. The conjugate action of $C \in \mathcal{C}_2$ on $\bar{P}_2$, denoted by $\Gamma_C$, is an automorphism of $\bar{P}_2$ such that $\Gamma_C( u \otimes v) \eqdef C (u \otimes v) C^\dagger $. When no confusion is possible, we shall simply denote $\Gamma_C$ by $\Gamma$.
%
%
%
\subsection{Quantum polarization for Pauli channels} 
\begin{definition}[Classical counterpart of a Pauli  channel]~\label{def:classical_channel}
Let $\mathcal{N}$ be a Pauli channel, that is, $\mathcal{N}(\rho) = \sum_u p_u u \rho u$, such that $u \in \bar{P}_1$, and $p_u \geq 0$ satisfying $\sum_u p_u = 1$. The classical counterpart of $\mathcal{N}$, denoted by $\mathcal{N}^\#$, is a classical channel from $\bar{P}_1$ (input alphabet) to $\bar{P}_1$ (output alphabet), which is defined by the transition probabilities,  ${\cal N}^\#(u \mid v) = p_w$, where $w \in \bar{P}_1$ is such that $u  v = w $.

\end{definition}

\begin{definition} [Classical mixture of Pauli (CMP) channels ]~\label{def:classical_channel_cmp}
A Classical Mixture of Pauli (CMP) channels is a quantum channel defined as, $\mathcal{N} (\rho) =  \sum_x \lambda_x \ket{x}\bra{x} \otimes \mathcal{N}_x(\rho)$, where $\{\ket{x} | x \in X\}$ is some orthonormal basis of an auxiliary system, $\mathcal{N}_x$ are Pauli channels, and $\lambda_x$ is a probability distribution on $X$.
\end{definition}
\noindent The definition of the classical counterpart channel from Definition~\ref{def:classical_channel} can be extended to the CMP channel by defining the classical counterpart ${\cal N}^\#$ as the mixture of classical channels ${\cal N}_x^\#$, where ${\cal N}_x^\#$ is used with probability $\lambda_x$. Hence, the input and output alphabets of ${\cal N}^\#$  are $\bar{P}_1$ and $X \times \bar{P}_1$, respectively, and the transition probability is given by ${\cal N}^\#(x, u \mid v) = \lambda_x\,{\cal N}^\#_x(u \mid v)$, for any $x \in X$, and $u, v \in \bar{P}_1$.

\begin{definition} [Equivalent classical channels]~\label{def:equivalent_cc}
Given two classical channels ${\cal U}$ and ${\cal V}$, we say they are equivalent and denote it by ${\cal U} \equiv {\cal V}$, if they are defined by the identical transition probability matrix up to a permutation of rows and columns.
\end{definition}
We now consider the channel combining and splitting procedure from~\cite{DGMS19}, on two copies of a quantum channel $\mathcal{W}_{A' \to B}$, where $A'$ and $B$ are the input and output quantum systems, respectively. Two instances of $\mathcal{W}_{A' \rightarrow B}$ are first combined using a two-qubit Clifford unitary $C \in \mathcal{C}_2$ as follows
\begin{equation}
(\mathcal{W} \bowtie_C \mathcal{W})( \rho_1 \otimes \rho_2 ) = \mathcal{W}_{A_1' \to B_1} \otimes \mathcal{W}_{A_2' \to B_2} \left( C( \rho_1 \otimes \rho_2) C^\dagger ) \right).
\end{equation}
The combined channel $\mathcal{W} \bowtie_C \mathcal{W}$ is then split into two quantum virtual channels, the bad channel $\mathcal{W} \boxast_C \mathcal{W}$ and the good channel $\mathcal{W} \varoast_C \mathcal{W}$, as follows
\begin{align}
(\mathcal{W} \boxast_C \mathcal{W})_{A_1' \to B_1B_2}(\rho) &= \mathcal{W}_{A_1' \to B_1} \otimes \mathcal{W}_{A_2' \to B_2} \left(C \left(\rho \otimes \frac{\ident}{2}\right) C^\dagger \right). \label{eq:quantum_split_bad}\\
(\mathcal{W} \varoast_C \mathcal{W})_{A_2' \to R_1B_1B_2} (\rho) &= \mathcal{W}_{A_1' \to B_1} \otimes \mathcal{W}_{A_2' \to B_2} \left(C \left( \Phi_{R_1 A_1'} \otimes \rho \right) C^\dagger \right). \label{eq:quantum_split_good}
\end{align}
Quantum polar code construction is obtained by recursively applying the above channel combining and splitting procedure on $N \eqdef 2^n$ copies of the quantum channel $\mathcal{W}$, with $n > 0$, which synthesizes $2^n$ quantum  virtual channels, $\mathcal{W}^{i_1 \cdots i_n}$, with $\{i_1 \cdots i_n\} \in \{0,1\}^n$~\cite{DGMS19} (see also Section~\ref{sec:Qpolar_constr}).
\medskip

When $\mathcal{W}$ is a CMP channel, it is shown in~\cite{DGMS19} that the synthesized virtual channels $\mathcal{W}^{i_1 \cdots i_n}$  are also CMP channels. Therefore, one can define classical counterpart channel for $\mathcal{W}^{i_1 \cdots i_n}$, which is denoted by ${\mathcal{W}^{i_1 \cdots i_n}}^\#$. 

Moreover, the classical channel combining and splitting procedure is defined for two copies of $\mathcal{W}^\#$, the classical counterpart of the CMP channel $\mathcal{W}$, using the permutation $\Gamma \eqdef \Gamma_C$, as channel combining operation. The channel combining in this case is given by
\begin{equation}
(\mathcal{W}^\# \bowtie_\Gamma \mathcal{W}^\#)( y_1, y_2| u, v ) = {\mathcal{W}^\#}^2 (y_1, y_2 | \Gamma(u, v)),
\end{equation}
 where $u, v \in \bar{P}_1$. The channel splitting yields the bad channel $\mathcal{W}^\# \boxast_\Gamma \mathcal{W}^\#$, and the good channel $\mathcal{W}^\# \varoast_\Gamma \mathcal{W}^\#$, as follows
\begin{align}
\mathcal{W}^\# \boxast_\Gamma \mathcal{W}^\#(y_1, y_2|u) = \sum_{v} \frac{1}{4} (\mathcal{W}^\# \bowtie_\Gamma \mathcal{W}^\#)( y_1, y_2| u, v ),  \label{eq:classical_com_1} \\
\mathcal{W}^\# \varoast_\Gamma \mathcal{W}^\#(y_1, y_2, u|v) =  \frac{1}{4} (\mathcal{W}^\# \bowtie_\Gamma \mathcal{W}^\#)( y_1, y_2| u, v ), \label{eq:classical_com_2}
\end{align}
Once again, by applying the above channel combining and splitting recursively on $2^n$ copies of the classical channel $\mathcal{W}^\#$, we obtain $2^n$ classical virtual channels, ${\mathcal{W}^\#}^{i_1 \cdots i_n}$, with $\{i_1 \cdots i_n\} \in \{0,1\}^n$. 

\medskip It is proven in~\cite{DGMS19} that classical channels  ${\mathcal{W}^{i_1 \cdots i_n}}^\#$ and ${\mathcal{W}^\#}^{i_1 \cdots i_n}$ are equivalent in the sense of Definition~\ref{def:equivalent_cc}, \emph{i.e.},
\begin{equation}
 {\mathcal{W}^{i_1 \cdots i_n}}^\# \equiv {\mathcal{W}^\#}^{i_1 \cdots i_n}.
\end{equation}
The above equation implies that $\mathcal{W}$ and $\mathcal{W}^\#$ polarize simultaneously under their respective polar code constructions (see Proposition 20 and Corollary 21 in~\cite{DGMS19}). Therefore, it would be sufficient to prove that polarization happens for any one of the two polar code constructions, as this would imply the same for the remaining one. In this work, we shall consider the polar code construction on the classical counterpart $\mathcal{W}^\#$ to show the multilevel polarization.

\subsection{Symmetric mutual information and Bhattacharyya parameter} \label{sec:definitions}
%
From now on, we denote $W \eqdef \mathcal{W}^\# $ for the sake of clarity. Recall that $W$ is a classical channel with the input alphabet $\bar{P}_1$. Note that $\bar{P}_1$ is isomorphic to the additive group $(\{00, 01, 10, 11\}, \oplus)$, where $\oplus$ denotes bitwise sum modulo 2. Throughout this paper, we shall identify $I \equiv 00$, $Z \equiv 01$, $X \equiv 10$, and $Y \equiv 11$. Using this identification, we may write $\bar{P}_1 = \{00, 01, 10, 11 \}$, or sometimes $\bar{P}_1 = \{0, 1, 2, 3 \}$, the notation will be clear from the context. 

\medskip \noindent We will use the symmetric mutual information of $W$, which is given by

\begin{equation} \label{eq:sym_mut_info}
I(W) = \frac{1}{4} \sum\limits_y \sum_{x \in \bar{P}_1} W(y|x) \text{ log}_2 \frac{W(y|x)}{P(y)},
\end{equation}
where $P(y) = \frac{1}{4} \sum_{x' \in \bar{P}_1} W(y|x')$. For any $x,x',d \in \bar{P}_1$, we further define two information measures $I(W_{x,x'})$ and $I_d(W)$ as follows
\begin{align}
&I(W_{x,x'}) =  \sum_y \frac{1}{2} \Big[ W(y|x) \text{ log}_2 \frac{W(y|x)}{\frac{1}{2}[W(y|x) + W(y|x')]} +  W(y|x') \text{ log}_2 \frac{W(y|x')}{\frac{1}{2}[W(y|x) + W(y|x')]}\Big]. \\
& I_d(W) = \frac{1}{4} \sum_x I(W_{x,x \oplus d}) \label{eq:information_I_d}.
\end{align}
Note that $I(W_{x,x'})$ is the symmetric mutual information of the binary-input channel obtained by restricting the input alphabet  of $W$ to $\{x,x'\} \subseteq \bar{P}_1$.

\medskip\noindent For $x,x', x'', d\in \bar{P}_1$, we define
\begin{align}
Z(W_{x, x^\prime}) &:= \sum_{y} \sqrt{W(y|x) W(y|x^\prime)}. \label{eq:bhattacharyya1} \\
Z_d(W) &:= \frac{1}{4}\sum_{x \in \bar{P}_1} Z(W_{x, x\oplus d}), \label{eq:def_Z_d}\\
       & = \frac{1}{2} [Z(W_{0, d}) + Z(W_{x'', x'' \oplus d})], \text { for any } x'' \neq 0,d, \label{eq:bhattacharyya_d}
\end{align}
where~(\ref{eq:bhattacharyya_d}) follows from $Z(W_{x,x'}) = Z(W_{x',x})$. Also, $Z(W_{x,x}) = 1$, $\forall x \in \bar{P}_1$, therefore $Z_d(W) = 1$ for $d = 0$. The Bhattacharyya parameter of the non-binary input channel $W$ is given by~\cite{sta09},
\begin{equation} \label{eq:bhattacharyya}
 Z(W) \eqdef \frac{1}{12} \sum_{x, x' \in \bar{P}_1: x \neq x^\prime}  Z(W_{x, x^\prime}) = \frac{1}{3} \sum_{d \in \bar{P}_1: d \neq 0} Z_d(W).
 \end{equation}
From~\cite{sta09}, we have the following relation between $I(W)$ and $Z(W)$, 
\begin{align}
I(W) &\geq \text{log}_2 \frac{4}{1 + 3Z(W)}.  \label{eq:I(W)_lower_bound}\\
I(W) &\leq  6 (\text{log}_2e) \sqrt{1 - Z(W)^2} \label{eq:I(W)_upper_bound}.
\end{align}
The first inequality from the above implies that $I(W)$ goes to $2$ if $Z(W)$ goes to $0$, and the second inequality implies that $I(W)$ goes to $0$ if $Z(W)$ goes to $1$.

%
%
%
%
%

\section{Noiseless, half-noisy and noisy channels} \label{sec:noisy_half_noiseless}

 In the lemma below, we show that if any two parameters from the set \break $\{Z_1(W), Z_2(W), Z_3(W)\}$, defined in~(\ref{eq:bhattacharyya_d}), approach 1, the remaining third parameter will also approach 1.
\begin{lemma} \label{lem:two_1_imp_thre_1}
 For any $\{d_1, d_2, d_3\} = \{ 1, 2, 3\}$, if $Z_{d_1}(W) \geq 1 - \epsilon_1$, and $Z_{d_2}(W) \geq 1 - \epsilon_2$, then,
 \begin{equation}
 Z_{d_3}(W) \geq 1 - \epsilon_3, \text{ where } \epsilon_3 = 4(\sqrt{\epsilon_1} + \sqrt{\epsilon_2})^2.
 \end{equation}
\end{lemma}

\begin{proof}
For $x \in \bar{P}_1$, consider a vector $\vec{A}(x)$ such that $\vec{A}(x) = ( \sqrt{W(y|x)}, y \in Y )$. It follows that $ \vec{A}(x) \cdot \vec{A}(x') = Z(W_{x,x'}) $ and $|\vec{A}(x) - \vec{A}(x')| = \sqrt{2 \big(1 - Z(W_{x,x'})\big)}$, where $|\vec{A}(x) - \vec{A}(x')|$ is the Euclidean distance between the vectors $\vec{A}(x)$ and $\vec{A}(x')$. Using the triangle inequality and $d_1 \oplus d_2 = d_3$, we have that

\begin{equation} \label{eq:trngle_bhatt}
\sqrt{ \big(1 - Z(W_{x,x \oplus d_3})\big)} \leq  \sqrt{ \big(1 - Z(W_{x,x \oplus d_1})\big)} + \sqrt{\big(1 - Z(W_{x \oplus d_1, x \oplus d_1 \oplus d_2 })\big)}.
\end{equation}

\noindent For $d \in \{d_1, d_2\}$, we have that $Z_d(W) \geq 1 - \epsilon \implies(1 - Z(W_{x, x\oplus d})) \leq 4 \epsilon, \forall x$. Then, from~(\ref{eq:trngle_bhatt}),
\begin{align}
&\sqrt{ \big(1 - Z(W_{x,x \oplus d_3})\big)}   \leq 2 (\sqrt{\epsilon_1} + \sqrt{\epsilon_2}), \forall x \nonumber \\
&\implies   Z_{d_3}(W)  \geq 1 - 4(\sqrt{\epsilon_1} + \sqrt{\epsilon_2})^2. \nonumber \qedhere
\end{align}
\end{proof}

We now define the partial channels of the non-binary input channel $W$.

\begin{definition}
(Partial channels). Consider $x  = x_1 x_2 \in \bar{P}_1 = \{00, 01, 10, 11\}$ is given as the channel input of $W$. We define the following three binary-input channels that are obtained by randomizing one bit of information from $x$,
{\small
\begin{align}
& W^{[1]}: x_1 \to y; \:W^{[1]}(y|0) = \frac{W(y|00) + W(y|01)}{2}, \: W^{[1]}(y|1) = \frac{W(y|10) + W(y|11)}{2}. \label{eq:partial_1} \\
& W^{[2]}: x_2 \to y; \: W^{[2]}(y|0) = \frac{W(y|00) + W(y|10)}{2}, \: W^{[2]}(y|1) = \frac{W(y|01) + W(y|11)}{2}. \label{eq:partial_2} \\
& W^{[3]}: x_1 \oplus x_2 \to y; \: W^{[3]}(y|0) = \frac{W(y|00) + W(y|11)}{2}, \: W^{[3]}(y|1) = \frac{W(y|01) + W(y|10)}{2}. \label{eq:partial_3}
\end{align}
}
\end{definition}

\noindent In particular, the partial channel $W^{[1]}$ takes $x_1$ as input and randomizes $x_2$, the partial channel $W^{[2]}$ takes $x_2$ as input and randomizes $x_1$, and the partial channel $W^{[3]}$ takes $x_1 \oplus x_2$ as input and randomizes both $x_1$ and $x_2$, individually. For $ \{d_1, d_2 , d_3\} = \{1, 2, 3\} $, the above three definitions can be merged into the following
\begin{equation} \label{eq:def:partial_w_d}
W^{[d_1]}(y|0) = \frac{W(y|0) + W(y| d_1 )}{2}, \text{ and } W^{[d_1]}(y|1) = \frac{W(y|d_2) + W(y| d_3)}{2}.
\end{equation}
\medskip

We now prove several bounds relating $Z_d(W)$, $Z(W^{[d]})$ and $I(W)$.
\begin{lemma} \label{lem:inequal}
Given $\{d_1, d_2, d_3\} = \{1, 2, 3\}$, we have the following inequalities, which bear similarities to Lemmas 9 and 10 from \cite{PB12}:
\begin{enumerate}
 \item  $Z(W^{[d_1]})  \leq  Z_{d_2}(W) + Z_{d_3} (W)$.
 
 \item $Z(W^{[d_1]})  \geq  Z_{d_i}(W)$, where $Z_{d_i}(W) = \text{\rm{max}} (Z_{d_2}(W), Z_{d_3}(W))$.
 
 \item $I(W)  \leq  \frac{1}{3} \sum_{d \in \{1, 2, 3\}} \sqrt{1 - Z_d(W)^2} + \frac{1}{3} \sum_{d \in \{1, 2, 3\}} \sqrt{1 - Z(W^{[d]})^2}$.
  
\end{enumerate}
\end{lemma}

\begin{proof}
Proof is given in Appendix~\ref{sec:proof-inequality}.
\end{proof}
\begin{lemma} \label{lem:two_0_I(W)_I(W_3)}
Given $Z_{d_1}(W) \leq \epsilon$, $Z_{d_2}(W) \leq \epsilon$, and $Z_{d_3}(W) \geq 1-\epsilon$, with $\epsilon > 0$, and $\{d_1, d_2, d_3\} = \{1, 2, 3\}$, then
\begin{itemize}

\item[(i)] $I(W^{[d_3]}) \in  [1 - \text{ \rm{log}}_2(1 + 2\epsilon), 1] $.

\item[(ii)] $|I(W) - I(W^{[d_3]})|  \leq \Delta$, \text{ where } $\Delta = \sqrt{2\epsilon} + \text{ \rm{log}}_2(1 + 2\epsilon)$.

\end{itemize} 
\end{lemma}

\begin{proof}

\textbf{Point $(i)$}: Since $W^{[d_3]}$ is a binary-input channel, $I(W^{[d_3]}) \leq 1$. From point $1$ of Lemma~\ref{lem:inequal}, we have that 
\begin{align} \label{eq:Z_d3}
0 \leq Z(W^{[d_3]})  \leq Z_{d_1}(W) + Z_{d_2}(W) \leq 2\epsilon.
\end{align}

\noindent Using the inequality $I(W_b) \geq 1 - \text{ log}_2(1 + Z(W_b))$ for any binary-input channel $W_b$ from~\cite{arikan09}, we can lower bound $I(W^{[d_3]})$ as follows

\begin{equation} \label{eq:cap_W^3}
I(W^{[d_3]}) \geq 1 - \text{ log}_2(1 + 2\epsilon).
\end{equation}
Hence,
\begin{equation} \label{eq:interval_IW3}
 I(W^{[d_3]}) \in [1 - \text{ log}_2(1 + 2\epsilon), 1].
\end{equation}

\medskip \noindent \textbf{Point $(ii)$}: From point $2$ of Lemma~\ref{lem:inequal}, we have that
\begin{equation} \label{eq:Z_d1_2}
Z(W^{[d_i]}) \geq Z_{d_3}(W) \geq 1 - \epsilon,  \forall d_i = d_1, d_2.
\end{equation}
In point $3$ of Lemma~\ref{lem:inequal}, substituting the lower bound on $Z_d(W)$, \emph{i.e.}, $Z_{d_1}(W) = Z_{d_2}(W) = 0, Z_{d_3}(W) = 1 - \epsilon$, and the lower bound on $Z(W^{[d]})$ from~(\ref{eq:Z_d3}) and~(\ref{eq:Z_d1_2}), \emph{i.e.},  $Z(W^{[d_1]}) = Z(W^{[d_2]}) = 1 - \epsilon, Z(W^{[d_3]}) = 0$, we have the following upper bound on $I(W)$,
\begin{equation} \label{eq:upper_bound_IW}
I(W) \leq 1 + \sqrt{2\epsilon}.
\end{equation}

\noindent From inequality in~(\ref{eq:I(W)_lower_bound}),  $I(W)$ can also be lower bounded, as below
\begin{align}
I(W) &\geq \text{ log}_2 \frac{4}{2 + \epsilon} \nonumber \\
     & \geq 1 - \text{ log}_2 ( 1 + \frac{\epsilon}{2} ). \label{eq:lower_bound_IW}
\end{align}
From~(\ref{eq:upper_bound_IW}) and~(\ref{eq:lower_bound_IW}), we have that 
\begin{equation} \label{eq:interval_IW}
I(W) \in [1 - \text{ log}_2 ( 1 + \frac{\epsilon}{2}), 1 + \sqrt{2\epsilon}].
\end{equation}
From~(\ref{eq:interval_IW3}) and~(\ref{eq:interval_IW}), we have that
\begin{align}
|I(W) - I(W^{[d_3]})| & \leq \Delta,
\end{align}
where $\Delta = \text{max}\left(\sqrt{2\epsilon} + \text{ log}_2(1 + 2\epsilon),  \text{ log}_2 ( 1 + \frac{\epsilon}{2})) = \sqrt{2\epsilon} + \text{ log}_2(1 + 2\epsilon\right) $.
\end{proof}
 We are now in a position to define the noiseless, half-noisy, and noisy channels.
\begin{definition}
Given $\delta >0$, a channel $W$ is said to be:
\begin{itemize}
\item[(i)] $\delta$-noiseless if $Z_1(W) < \delta, Z_2(W) < \delta$, and $Z_3(W) < \delta$.

\item[(ii)] $\delta$-noisy if $Z_1(W) > 1-\delta$, and $Z_2(W) > 1-\delta$.

\item[(iii)]  $\delta$-half-noisy of type $d_3$, if  $Z_{d_1}(W) < \delta, Z_{d_2}(W) < \delta$, and $Z_{d_3}(W) > 1-\delta$, with $\{d_1, d_2, d_3\} = \{1, 2,3\}$.
 \end{itemize}
\end{definition}
\noindent Recall that $W$ takes as input two bits $x_1x_2$, where $x_1, x_2$, and $x_1 \oplus x_2$ are inputs to the partial channels $W^{[1]}, W^{[2]}$, and $W^{[3]}$, respectively.

\medskip \noindent If $W$ is such that $Z_1(W) < \delta, Z_2(W) < \delta, \text{ and }Z_3(W) < \delta$, using~(\ref{eq:I(W)_lower_bound}), we have that  $I(W) \to 2$ as $\delta \to 0$. Therefore, we call $W$, $\delta$-noiseless. 

\medskip \noindent If $W$ is such that $Z_1(W) > 1-\delta, Z_2(W) > 1-\delta$, using~(\ref{eq:I(W)_upper_bound}) and Lemma~\ref{lem:two_1_imp_thre_1}, we have that $I(W) \to 0$ as $\delta \to 0$.  Therefore, we call $W$, $\delta$-noisy.

\medskip \noindent If $W$ is such that $Z_{d_1}(W) \leq \delta$, $Z_{d_2}(W) \leq \delta$, and $Z_{d_3}(W) \geq 1 - \delta$, with $\{d_1, d_2, d_3\} = \{1, 2,3\}$ and $\delta \to 0$, from point $(i)$ of Lemma~\ref{lem:two_0_I(W)_I(W_3)}, the binary-input partial channel $W^{[d_3]}$ tends to be noiseless, that is, $I(W^{[d_3]}) \to 1$. We may take $d_3 = 1$, without loss of generality. Then, we can reliably transmit one bit of information, namely $x_1$, the input to the partial channel $W^{[1]}$, using $W$. Moreover, from point $(ii)$ of Lemma~\ref{lem:two_0_I(W)_I(W_3)}, $I(W^{[d_3]}) \to I(W)$. Thus, the remaining one bit from the input of $W$, namely $x_2$, the input to the partial channel $W^{[2]}$, is completely randomized or erased. Therefore, we call $W$, ``$\delta$-half-noisy of type $d_3$''.

\medskip

\section{Multilevel polarization} \label{sec:multilevel}

In this section, we show that the CMP channel polarize into noiseless, half-noisy or noisy channels, under the recursive channel combining and splitting procedure, using a fixed two qubit Clifford as channel combining operation. We take the following two-qubit gate as channel combining operation. 
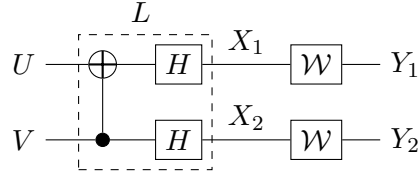
\begin{figure}[H]
\[
\begin{tikzpicture}
\draw
(0, 0) node[not](n){} 
(n)++(1, 0) node[draw](hd1){$H$} 
(hd1.east)++ (1.5, 0) node[draw](can1){$\mathcal{W}$} 
(can1.east)to ++(0.5,0)node[right]{$Y_1$}
(hd1.east) to node[above](){$X_1$} (can1)
(n) to ++(-0.75, 0) node[left](){$U$}
(n) to (hd1.west)
(n) ++ (0, -1) node[phase](c){} 
(c) ++ (1, 0)  node[draw](hd2){$H$} 
(c) to (hd2.west) 
(hd2.east) ++(1.5, 0) node[draw](can2){$\mathcal{W}$}
(hd2.east) to node[above](){$X_2$} (can2)
(can2.east)to ++(0.5,0) node[right]{$Y_2$}
(c) to ++(-0.75, 0) node[left](){$V$}
(c) to (n)
(hd1) ++ (-0.5, 0.7) node[]{$L$}
; 

\node[draw, dashed, fit=(hd1) (hd2) (c) (n)] (g) {};

\end{tikzpicture}
\]
\caption{ Two-qubit Clifford gate $L$. Here $H$ is the Hadamard gate.}
\label{fig:chan_com_oper}
\end{figure}
 The above two-qubit Clifford unitary $L$ generates the same permutation $\Gamma$ on $\bar{P}_1 \times \bar{P}_1$ as the Clifford $L_{3,3}$ from~\cite[Figure $4$]{DGMS19}. As only permutation $\Gamma$ matters for the polarization of Pauli channels, the gate $L$ is equivalent to $L_{3,3}$ for our purposes.  Note that the gate $L$ applies the same single qubit gate, namely the Hadamard gate $H$, on both qubits after the CNOT gate. Also, it is important to mention that multilevel polarization may not happen for all the Cliffords given in~\cite[Figure $4$]{DGMS19}.

\medskip
As mentioned before, we will use the polar code construction on the classical counterpart of the CMP channel, \emph{i.e.}, $W \eqdef \mathcal{W}^\#$, to prove the multilevel polarization. The channel combining operation for two copies of $W$ is $\Gamma(L)$, that is, the permutation generated by the conjugate action of $L$ on $\bar{P}_1 \times \bar{P}_1$, which is depicted in the following figure\footnote{Recall that $00 \equiv I, 01 \equiv Z, 10 \equiv X, 11 \equiv Y$.},

\begin{figure}[H]
\[
\begin{tikzpicture}
\draw
(0, 0) node[not](n){} 
(n)++(1, 0) node[draw](hd1){$H$} 
(hd1.east)++ (2.5, 0) node[draw](can1){$W$} 
(can1.east)to ++(0.5,0)node[right]{$y_1$}
(hd1.east) to node[above](){$u_2, u_1 \oplus v_1$} (can1)
(n) to ++(-0.75, 0) node[left](){$u_1, u_2$}
(n) to (hd1.west)
(n) ++ (0, -1) node[phase](c){} 
(c) ++ (1, 0)  node[draw](hd2){$H$} 
(c) to (hd2.west) 
(hd2.east) ++(2.5, 0) node[draw](can2){$W$}
(hd2.east) to node[above](){$u_2 \oplus v_2, v_1$} (can2)
(can2.east)to ++(0.5,0) node[right]{$y_2$}
(c) to ++(-0.75, 0) node[left](){$v_1, v_2$}
(c) to (n)
(hd1) ++ (-0.5, 0.7) node[]{$\Gamma(L)$}
; 
\node[draw, dashed, fit=(hd1) (hd2) (c) (n)] { };
\end{tikzpicture}
\]
\caption{ The permutation $\Gamma(L)$. To avoid any possible confusion, two bits of the input and output symbols are separated here by a comma.
}
\label{fig:perm_comb_2}
\end{figure}
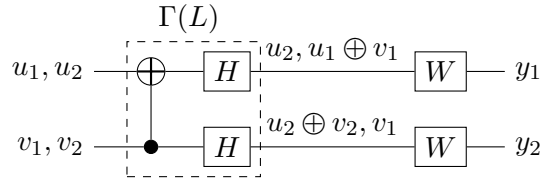
\noindent From~(\ref{eq:classical_com_1}) and~(\ref{eq:classical_com_2}), the virtual channels obtained after the channel combining and splitting procedure on two copies of $W$, using $\Gamma(L)$ as channel combining operation, are given by
\begin{align}
(W \boxast W)(y_1, y_2|u_1,u_2) &= \frac{1}{4} \sum_{v_1,v_2} W(y_1|u_2 , u_1 \oplus v_1) W(y_2|u_2 \oplus v_2 , v_1), \label{eq:bad_channel_W} \\
 (W \varoast W)(y_1, y_2, u_1,u_2|v_1,v_2) &= \frac{1}{4} W(y_1|u_2 , u_1 \oplus v_1) W(y_2|u_2 \oplus v_2 , v_1), \label{eq:good_channel_W}
\end{align}
where $u_1,u_2, v_1,v_2 \in \{0, 1\}$. From the chain rule of mutual information, we have that
\begin{equation} \label{eq:I_preserve}
I(W \boxast W) + I(W \varoast W) = 2 I(W),
\end{equation}
which means that the mutual information is preserved under the above channel combining and splitting procedure.

\medskip We now give the following two lemmas.
\begin{lemma} \label{lem:inequality_bht}
The following equalities hold for the good channel $W \varoast W$,
\begin{align}
Z_1( W \varoast W ) &= Z_2(W). \label{eq:ineq_Zgood_1} \\
Z_2(W \varoast W) &=Z_1(W)^2. \label{eq:ineq_Zgood_2}
\end{align}
\end{lemma}

\begin{proof}
Proof is given in Appendix~\ref{sec:proof_inequality_bht}.
\end{proof}

\begin{lemma} \label{lem:partial_canvirtual_inequality}
The following inequalities hold for the partial channels, $(W \boxast W)^{[i]}$ and $(W \varoast W)^{[i]}$, for all $ i \in \{1, 2\}$,
\begin{align}
Z\big((W \boxast W)^{[1]}\big) &\leq 2Z(W^{[2]}) - Z(W^{[2]})^2.  \label{eq::ineq_1}\\
Z\big((W \boxast W)^{[2]}\big) &= Z(W^{[1]}).  \label{eq::ineq_2} \\
Z\big((W \varoast W)^{[1]}\big) &= Z_1(W) Z(W^{[2]}). \label{eq::ineq_3}  \\
Z\big((W \varoast W)^{[2]}\big) &\leq Z(W^{[1]}).  \label{eq::ineq_4} 
\end{align}
\end{lemma}

\begin{proof}
Proof is given in Appendix~\ref{sec:proof_inequality_partial_bht}.
\end{proof}

\medskip\noindent We define $ W^0 \eqdef W \boxast W$ and  $ W^1 \eqdef W \varoast W $, and consider the recursive application of channel combining and splitting procedure,  $W \mapsto (W^0, W^1)$. After two steps of polarization, we have a set of four virtual channels, $(W^{i_1})^{i_2}, \forall i_1i_2 \in \{0, 1\}^2$. Similarly, after $n$ polarization steps, we have the following set of $2^n$ virtual channels,
\begin{equation} \label{eq:virtual_channels}
 W^{i_1 \cdots i_{n}} \eqdef (W^{i_1 \cdots i_{n-1}})^{i_n}, \forall i_1 \cdots i_n \in \{0, 1\}^n. 
\end{equation}

\medskip We now state the multilevel polarization theorem.
\begin{theorem} \label{thm:main_multi_level}
Let $\{W^{i_1 \cdots i_n} | i_1 \cdots i_n \in \{0, 1\}^n \}$ be the set of virtual channels defined in~(\ref{eq:virtual_channels}), when the permutation $\Gamma(L)$ is used as channel combining operation. Then, for any $\delta > 0$, 
{ \small $$\lim_{n \to \infty} \frac{\#\{i_1 \cdots i_n \in \{0, 1\}^n  \mid  W^{i_1 \cdots i_n} \text{ is either } \delta\text{-noiseless, } \delta\text{-half-noisy of type 1 or 2, or } \delta \text{-noisy} \}}{2^n} = 1.$$}
\end{theorem}
Note that it is sufficient to prove the above theorem assuming that $n$ goes to infinity through even values $2, 4, 6, \dots $. Indeed if the above theorem holds for $n$ going to infinity through even values, we can set $W = W^{i_1}$, for all $i_1 \in \{0, 1\}$, and then it follows that it also holds for $n$ going to infinity through odd values. Therefore, from now on, we  assume that $n = 2m$.

\medskip In~(\ref{eq::ineq_1})-(\ref{eq::ineq_4}), the upper bound on $Z({W^{i_1}}^{[d]})$, for any $i_1 \in \{0, 1\}$ and $d \in \{1, 2\}$, is a function of $Z(W^{[d']})$, such that $\{d,d'\} = \{1, 2\}$. Therefore, applying the transform $ W \to (W^0, W^1)$ twice, we get an upper bound on $Z({W^{i_1i_2}}^{[d]}), \forall i_1i_2 \in \{0, 1\}^2$, which is a function of $Z(W^{[d]})$. For this reason, it is convenient to consider even steps of polarization, \emph{i.e.}, $n =2m$, and use $W\to  (W^{00}, W^{01}, W^{10}, W^{11})$ as our basic transform for recursion. For any given sequence $i_1 \cdots i_n \in \{0, 1\}^n$, we write $i_1 \cdots i_n = \omega_1 \cdots \omega_m$, such that $\omega_k = i_{2k-1} i_{2k} \in \{0,1\}^2, \forall k > 0$. 

\medskip To prove Theorem~\ref{thm:main_multi_level}, we will express the limit therein as the probability of an event on a probability space. Therefore, suppose that $\{B_i: i= 0, 1, \dots \infty \}$ is a sequence of random i.i.d variables defined on a probability space $(\Omega, \mathcal{F}, P)$, where each $B_i$ takes values in $\{0, 1\}^2$ with equal probability, meaning that $P(B_i = 00) = P(B_i = 01) = P(B_ i= 10) = P(B_i = 11) =  \frac{1}{4}$. Let ${\cal F}_0 = \{ \phi, \Omega\}$ be the trivial $\sigma$-algebra and ${\cal F}_m$, $m\geq 1$ be the $\sigma$-field generated by $(B_1, \dots, B_m)$.
Define a random sequence of channels $\{W_m: m \geq 0\}$ on the probability space, such that $W_0 = W$, and at any time $m \geq 1$, $W_m = W_{m-1}^{\omega_m}$, where $\omega_m \in \{0, 1\}^2$ is the value of $B_m$. Therefore, if $B_1 = \omega_1, B_2 =\omega_2, \dots, B_m = \omega_m$, we have that  $W_m = W^{\omega_1 \cdots \omega_m}$.


\medskip\noindent For  any $0< \delta < \frac{1}{2}$, we define the following events on probability space,
\begin{align}
A &= \{\omega \in \Omega: \exists m_0, \forall  m \geq m_0, W_m \text{ is } \delta\text{-noiseless}  \}. \label{eq:event_A} \\
B &= \{\omega \in \Omega:  \exists  m_0, \forall m \geq m_0,  W_m \text{ is } \delta\text{-half-noisy of type } 1 \}. \label{eq:event_B} \\
C &= \{\omega \in \Omega:  \exists m_0, \forall m \geq m_0, W_m \text{ is } \delta\text{-half-noisy of type } 2 \}. \label{eq:event_c}\\
D &= \{\omega \in \Omega:  \exists m_0, \forall m \geq m_0, W_m \text{ is } \delta\text{-noisy} \}. \label{eq:event_D}
\end{align}

\noindent The intersection of any two of the above sets is the null set. Note that the limit in Theorem~\ref{thm:main_multi_level} is equal to $P(A \cup B \cup C \cup D)$, hence, in other words, Theorem~\ref{thm:main_multi_level} states that one of the events from $A, B, C, D$ occurs with probability 1, as $n$ goes to infinity. We first prove the following Lemmas~\ref{lem:pol_cond},~\ref{lem:supermartingale} and~\ref{lem:sub_set}, and then use them to prove the above polarization theorem.

\begin{lemma} \label{lem:pol_cond}
 Consider a stochastic process $\{T_m:m \geq 0\}$ defined on $(\Omega, \mathcal{F}, P)$ such that it satisfies the following properties:
\begin{enumerate}
\item $T_m$ takes values in $[0, 1]$ and is measurable with respect to $\mathcal{F}_m$, that is, $T_0$ is a constant and $T_m$ is a function of $(B_1, \dots, B_m)$.

\item Process $\{(T_m, \mathcal{F}_m): m \geq 0\}$ is a super-martingale.

\item $T_{m + 1} = T_m^2$ with probability $\frac{1}{2}$.
\end{enumerate}
Then, the limit $T_\infty = \lim_{m \rightarrow \infty} T_m$ exists with probability 1, and $T_\infty$ takes values in $\{0, 1\}$.

\begin{proof}
The proof is  similar to~\cite[Proposition 9]{arikan09}. Since the process $\{(T_m, \mathcal{F}_m): m \geq 0\}$ is a super-martingale, $T_m$ converges with probability 1. This gives the proof of the first part, which implies that $ \text{lim}_{m \to \infty} |T_{m+1} - T_{m}| = 0$. As $T_{m + 1} = T_m^2$ with probability $\frac{1}{2}$, it follows that $T_m$ takes values in $\{0, 1\}$.
%
\end{proof}

\end{lemma}
\begin{lemma} \label{lem:supermartingale}
For all d = 1, 2, the process $\{Z(W^{[d]}_m): m \geq 0\}$ defined on $(\Omega, \mathcal{F}, P)$, is a super-martingale and there exist $q_1 = q_1(d), q_2 = q_2(d) \in \{ 0, 1\}^2$, such that when $B_{m+1} \in \{q_1, q_2\}$, $Z(W^{[d]}_{m+1}) \leq Z(W^{[d]}_m)^2$.
\end{lemma}

\begin{proof}
For $d = 1$, using~(\ref{eq::ineq_1})-(\ref{eq::ineq_4}) with $W = W_m$, we get 
\begin{align}
Z({W_{m}^{00}}^{[1]}) & \leq 2Z({W_m^0}^{[2]}) - Z((W_m^0)^{[2]})^2 = 2Z(W_m^{[1]}) - Z(W_m^{[1]})^2, \label{eq:bad_bad_can1} \\
Z({W_m^{01}}^{[1]}) & \leq Z_1(W_m^0) Z({W_m^0}^{[2]}) \leq Z({W}_m^{[1]})^2, \label{eq:second} \\
Z({W_m^{10}}^{[1]}) & \leq 2Z({W_m^1}^{[2]}) - Z({W_m^1}^{[2]})^2 \leq 2Z({W}_m^{[1]}) - Z({W}_m^{[1]})^2, \\
Z({W_m^{11}}^{[1]}) & = Z_1(W_m^1) Z({W_m^1}^{[2]})) \leq Z_2(W_m) Z(W_m^{[1]}) \label{eq:good_good_can1},  
\end{align}

where the second inequality in~(\ref{eq:second})  uses the inequality $Z_1(W^0) \leq  Z({W^0}^{[2]})$ from [Lemma~\ref{lem:inequal}, point $2$], and second inequality in~(\ref{eq:good_good_can1}) uses $Z_1(W_m^1)=Z_2(W_m)$ from~(\ref{eq:ineq_Zgood_1}). From~(\ref{eq:bad_bad_can1})-(\ref{eq:good_good_can1}) and $Z_2(W) \leq  Z(W^{[1]})$ [Lemma~\ref{lem:inequal}, point $2$], it follows,
\begin{align}
 \sum_{i_1,i_2 \in \{0, 1\}} Z_1({W_m^{i_1i_2}}^{[1]}) & \leq 4 Z_1(W_m^{[1]}).
\end{align}
Hence, the process  $\{Z(W_m^{[1]}): m \geq 0\}$ is a super-martingale and also when $B_{m+1} \in \{01, 11\}$, we have that $Z(W^{[1]}_{m+1}) \leq Z(W^{[1]}_m)^2$.

\medskip\noindent For $d=2$, from~(\ref{eq::ineq_1})-(\ref{eq::ineq_4}) with $W = W_m$, we have that
\begin{align}
Z({W_m^{00}}^{[2]}) & = Z({W_m^0}^{[1]})) \leq 2Z(W_m^{[2]}) - Z(W^{[2]})^2. \label{eq:bad_bad_can2}\\
Z({W_m^{01}}^{[2]}) & \leq Z({W_m^0}^{[1]})  \leq 2Z(W_m^{[2]}) - Z(W^{[2]})^2.   \\
Z({W_m^{10}}^{[2]})) & = Z({W_m^1}^{[1]}) = Z_1(W_m) Z(W_m^{[2]}). \\
Z({W_m^{11}}^{[2]})) & \leq  Z({W_m^1}^{[1]}) = Z_1(W_m) Z(W_m^{[2]}).  \label{eq:good_good_can2}
\end{align}
From~(\ref{eq:bad_bad_can2})-(\ref{eq:good_good_can2}) and using $Z_1(W) \leq Z(W^{[2]})$ [Lemma~\ref{lem:inequal}, point $2$], we have that
\begin{equation}
 \sum_{i_1,i_2 \in \{0, 1\}} Z({W_m^{i_1i_2}}^{[2]})) \leq 4 Z(W_m^{[2]}).
\end{equation}
Thus, process  $\{Z(W^{[2]}_m): m \geq 0\}$ is a super-martingale, and also when $B_{m+1} \in \{10, 11\}$, we have that $Z(W^{[2]}_{m+1}) \leq Z(W^{[2]}_m)^2$.
\end{proof}

\begin{lemma} \label{lem:sub_set}
Define the following events for $d = 1, 2$,
\begin{align}
S^{[d]}(\delta) &\eqdef \{ \omega \in \Omega : \exists m_0,  \forall m \geq m_0, Z(W^{[d]}_m) < \delta \}. \label{eq:event_S^d}\\
T^{[d]}(\delta) &\eqdef \{ \omega \in \Omega :\exists m_0,  \forall m \geq m_0, Z(W^{[d]}_m) > 1 - \delta \}. \label{eq:event_T^d}\\
S_{d}(\delta) &\eqdef \{ \omega \in \Omega : \exists m_0, \forall m \geq m_0, Z_{d}(W_m) < \delta, Z_{3}(W_m) < \delta \}. \label{eq:event_S_d}\\
T_{d}(\delta) &\eqdef \{ \omega \in \Omega : \exists m_0,  \forall m \geq m_0, Z_{d}(W_m) > 1 - \delta \}. \label{eq:event_T_d}
\end{align}
Then,
\begin{itemize}
\item[(i)] $P(S^{[d]}(\delta) \cup T^{[d]}(\delta)) = 1$, $\forall  d = 1, 2$.

\item[(ii)] Given $\{d,d'\} = \{1, 2\}$, then 

\begin{enumerate}
\item[(a)] $S^{[d]}(\delta) \subseteq S_{d'}(\delta)$.
\item[(b)] $T^{[d]}(\delta)  \subseteq T_{d'}(\delta) \text{ with probability 1}$.
\end{enumerate}
\end{itemize}
\end{lemma}

\begin{proof}

\noindent \textbf{Point $(i)$}: It follows directly from Lemmas~\ref{lem:pol_cond} and~\ref{lem:supermartingale}. As a consequence, note that any $\omega  \in \Omega$ belongs to one of the sets, $S^{[1]}(\delta) \cap S^{[2]}(\delta)$, $S^{[1]}(\delta) \cap T^{[2]}(\delta)$, $T^{[1]}(\delta) \cap S^{[2]}(\delta)$ and $T^{[1]}(\delta) \cap T^{[2]}(\delta)$ with probability $1$.  This will be used in the proof of Theorem~\ref{thm:main_multi_level}.

\medskip \noindent \textbf{Point $(ii).(a)$}: From [lemma~\ref{lem:inequal}, point $2$], we have that $Z_{d'}(W_m) \leq Z(W_m^{[d]}) $  and  $Z_{3}(W_m) \leq Z(W_m^{[d]}) $, for $\{d, d'\} = \{1, 2\}$. Then, it immediately follows by definitions of $S^{[d]}(\delta)$ and $S_{d'}(\delta)$ that $S^{[d]}(\delta) \subseteq S_{d'}(\delta)$.


\medskip \noindent \textbf{Point $(ii).(b)$}\footnote{Note that $T_{d'}(\delta) \subseteq T^{[d]}(\delta)$, by the same reasoning as in the proof of previous point $(ii).(a)$. Hence,  point $(ii).(b)$ actually implies that $T^{[d]}(\delta) = T_{d'}(\delta)$ with probability 1.}:
We assume $T^{[d]}(\delta) \not \subset T_{d'}(\delta)$ with non-zero probability and disprove it by contradiction. The above  assumption implies that the following event,
\begin{equation}
E = \{ \omega \in \Omega: \omega \in T^{[d]}(\delta), \omega \not\in T_{d'}(\delta)  \},
\end{equation}
occurs with non-zero probability, that is, $P(E) > 0$.

\medskip\noindent Define an event $\mathcal{E}_m$ such that  $Z_{d'}(W_m) \leq 1 - \delta$, that is, given $B_1 = \omega_1, \dots, B_m = \omega_m$, we have $Z_{d'}(W^{\omega_1 \cdots \omega_m}) \leq 1 - \delta$. Any $\omega \in E$ belongs to infinitely many $\mathcal{E}_m$  because if there exists a $m_0$ such that $Z_{d'}(W_m) > 1 - \delta$, for all $m > m_0$, this would imply $\omega \in T_{d'}(\delta)$, which is not true by assumption. Given  $\omega \in E$, consider $M = \{ m_1, m_2, \dots \}$ as the set of instances such that for all $m_i \in M$, $\omega \in \mathcal{E}_{m_i}$. Further, take $m$ such that  $B_{m+1}  =  B_{m+2}= 11$ happens, probability of such an event is given by $P(B_{m+1} =  B_{m+2} = 11)  = \frac{1}{16} > 0$, for any $m\geq 1$, therefore, $\sum_{m_i \in M} P(B_{m_i+1} = B_{m_i+2} = 11) = \infty$. Since $\{B_{m}:m \geq 1\}$ are i.i.d. random variables, using Borel-Cantelli lemma, there are infinitely many $m_i \in M$ for which $B_{m_i+1} =  B_{m_i+2} =11$.

\medskip \noindent The condition $\omega \in T^{[d]}(\delta)$ implies that $Z(W^{[d]}_m) > 1 - \delta$, for all $m \geq m_0$. Take a $m \geq m_0$ such that $\omega \in \mathcal{E}_m$, and $B_{m+1} = B_{m+2} =  11$. Then, we have the following for all $d = 1, 2$,
\begin{align*}
Z(W^{[d]}_{m+2}) & \leq Z_{d'}(W_{m+1}) Z(W^{[d]}_{m+1}) \nonumber \\
                 & \leq Z_{d'}(W_{m})^2 Z_{d'}(W_{m}) Z_d(W^{[d]}_{m}) \nonumber \\
& \leq (1-\delta)^3 < (1 - \delta), \nonumber
\end{align*}
where both the first and second inequalities use~(\ref{eq:good_good_can1}) and~(\ref{eq:good_good_can2}), the second inequality also uses $Z_{d}(W_{m+1}) = Z_{d}(W_{m})^2 $ (from~(\ref{eq:ineq_Zgood_1}) and~(\ref{eq:ineq_Zgood_2})), and the third inequality follows from the assumption that $Z_{d'}(W_{m}) \leq (1 - \delta)$. Hence, we have a contradiction with the statement that $Z(W^{[d]}_m) > 1 - \delta$ for all $m \geq m_0$. Therefore, $T^{[d]}(\delta) \subseteq T_{d'}(\delta)$ holds with probability $1$.

%
%
%
%
\end{proof}
\noindent \textbf{Proof of Theorem~\ref{thm:main_multi_level}}:
We have the following by definition,
\begin{align}
S_1(\delta) \cap S_2(\delta) &= A. \nonumber \\
T_1(\delta) \cap S_2(\delta) &=  B. \nonumber \\
S_1(\delta) \cap T_2(\delta) &=  C. \nonumber \\
T_1(\delta) \cap T_2(\delta) & = D. \nonumber 
\end{align}

\noindent From point $(ii).(a)$ of Lemma~\ref{lem:sub_set}, we have that $S^{[1]}(\delta) \cap S^{[2]}(\delta) \subseteq S_1(\delta) \cap S_2(\delta)$, which means $\omega \in S^{[1]} \cap S^{[2]} \implies \omega \in S_1(\delta) \cap S_2(\delta)$. Similarly, from point $(ii).(a)$ and point $(ii).(b)$ of Lemma~\ref{lem:sub_set}, we have that
\begin{align*}
S^{[1]}(\delta) \cap T^{[2]}(\delta) &\subset T_1(\delta) \cap S_2(\delta). \nonumber \\
T^{[1]}(\delta) \cap S^{[2]}(\delta) &\subset S_1(\delta) \cap T_2(\delta). \nonumber \\
T^{[1]}(\delta) \cap T^{[2]}(\delta) &\subset T_1(\delta) \cap T_2(\delta). \nonumber 
\end{align*}
From point $(i)$ of Lemma~\ref{lem:sub_set}, we know that one of the events from $S^{[1]}(\delta) \cap S^{[2]}(\delta)$, $S^{[1]}(\delta) \cap T^{[2]}(\delta)$, $T^{[1]}(\delta) \cap S^{[2]}(\delta)$ and $T^{[1]}(\delta) \cap T^{[2]}(\delta)$ happens with probability $1$, therefore, we have that
\begin{align}
S^{[1]}(\delta) \cap S^{[2]}(\delta) &= S_1(\delta) \cap S_2(\delta)=A, \label{eq:pol_1} \\
S^{[1]}(\delta) \cap T^{[2]}(\delta) &= T_1(\delta) \cap S_2(\delta)=B, \\
T^{[1]}(\delta) \cap S^{[2]}(\delta) &= S_1(\delta) \cap T_2(\delta)=C,  \\
T^{[1]}(\delta) \cap T^{[2]}(\delta) &= T_1(\delta) \cap T_2(\delta)=D, \label{eq:pol_4}  
\end{align}
Hence, $P(A \cup B \cup C \cup D) = 1$. 
\ \hfill $\qed$
%
%

\section{Quantum coding scheme} \label{sec:Qpolar_constr}
In this section, we propose a quantum polar coding scheme for CMP channels based on multilevel channel polarization proven in the previous section.

\subsection{Code construction}
The construction for quantum polar code is illustrated in Figure~\ref{fig:quantum_pol_constr}, for $N = 2^3$, where the virtual channels are written above the wires corresponding to their channel inputs. Similar to~\cite{arikan09}, to construct a quantum polar code of length $N = 2^n$, we start with $N$ copies of a CMP channel $\mathcal{W}$, and divide them into $\frac{N}{2}$ pairs. Then, the channel combining and splitting procedure is applied on every pair, using our two-qubit Clifford unitary $L$ as channel combining operation, which gives $\frac{N}{2}$ copies of both $\mathcal{W}^0 \eqdef \mathcal{W} \boxast \mathcal{W}$ and $\mathcal{W}^1 \eqdef \mathcal{W} \varoast \mathcal{W}$. In the next step, for all $ i_1 \in \{0, 1\}$, we regroup $\frac{N}{2}$ copies of each $W^{i_1}$, divide them in $\frac{N}{4}$ pairs, and once again apply the channel combining and splitting procedure on each pair, which gives $\frac{N}{4}$ copies of $(\mathcal{W}^{i_1})^{i_2}$ for all $i_1i_2 \in \{0, 1\}^2$. Basically, after each step, we regroup the same copies of a virtual channel, divide them into pairs, and apply the channel combining and splitting procedure on each pair. Repeating the above procedure for $n$ steps, we get $2^n$ virtual channels denoted by $\mathcal{W}^{i_1 \cdots i_{n}}$, where $i_1 \cdots i_{n} \in \{0, 1\}^n$, and the procedure stops after $n$ steps as only one copy of each $\mathcal{W}^{i_1 \cdots i_{n}}$ is available.

\medskip\noindent We also consider the same procedure as above on $N$ copies of the classical counterpart channel $\mathcal{W}^\#$, using permutation $\Gamma(L)$ as channel combining operation. This will provide a classical polar code construction, which synthesizes $2^n$ virtual channels ${\mathcal{W}^\#}^{i_1 \cdots i_n}$ for $i_1 \cdots i_n \in \{0, 1\}$. As explained in Section~\ref{sec:preliminary} (see also ~\cite[Proposition 20 and Corollary 21]{DGMS19}), there is one to one correspondence between $\mathcal{W}^{i_1 \cdots i_n}$ and ${\mathcal{W}^\#}^{i_1 \cdots i_n}$ in the sense that ${\mathcal{W}^{i_1 \cdots i_n}}^\# \equiv {\mathcal{W^\#}}^{i_1 \cdots i_n}$.

\begin{figure}[H]
\centering
\includegraphics[scale=0.6]{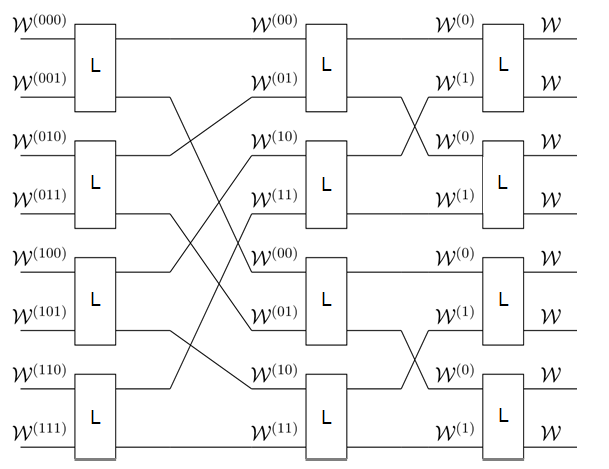}
\caption{Quantum Polar code construction for $N=2^3$. Here, $L$ is the two-qubit Clifford gate from Figure~\ref{fig:chan_com_oper}.}
\label{fig:quantum_pol_constr}
\end{figure}



\subsection{Encoding}
Consider $n$ steps of polarization with $n > 0$. As mentioned before, the polar code construction synthesizes $N = 2^{n}$ virtual channels corresponding to each $ i \in \{0, 1, \dots, N-1\}$. We shall denote, $\mathcal{W}^{(i)} \eqdef \mathcal{W}^{i_1 \cdots i_{n}}$, where $i_1 \cdots i_{n}$ is the binary representation of $i \in \{0, 1, \dots, N-1\}$. Similar to Section~\ref{sec:multilevel}, we define the following sets,
\begin{align}
 \mathcal{A} &= \{i \in  \{0, 1, \dots, N-1\}: W^{(i)} \text{ is } \delta\text{-noiseless}  \}. \label{eq:setA}\\
 \mathcal{B} &= \{ i \in  \{0, 1, \dots, N-1\} : W^{(i)} \text{ is } \delta\text{-half-noisy of type } 1 \} \label{eq:setB}. \\
\mathcal{C} &= \{i \in  \{0, 1, \dots, N-1\} : W^{(i)} \text{ is } \delta\text{-half-noisy of type } 2 \}. \label{eq:setC} \\
 \mathcal{D} &= \{i \in  \{0, 1, \dots, N-1\} : W^{(i)} \text{ is } \delta\text{-noisy} \}.  \label{eq:setD}
\end{align}
From Theorem~\ref{thm:main_multi_level}, it follows that for sufficiently large $N$, all but a vanishing fraction of  elements from the set $\{0, 1, \dots, N-1\}$ belong to one of the above sets. Let $\bar{\mathcal{D}}$ denote the complement of $\mathcal{A} \cup \mathcal{B} \cup \mathcal{C}$.  The inputs to the virtual channels $\mathcal{W}^{(i)}$ are supplied as follows for $i$ in $\mathcal{A}$, $\mathcal{B}$, $\mathcal{C}$ and $\bar{\mathcal{D}}$, 
\begin{itemize}
\item If $ a \in \mathcal{A}$, the corresponding $\mathcal{W}^{(a)}$ is used for quantum communication.
\item If  $b \in \mathcal{B}$, the input of the corresponding $\mathcal{W}^{(b)}$ is set to $\ket{+}$, the eigenstate of the Pauli $X$ operator with eigenvalue 1 (since $Z$ part of the input $x_1x_2 \in \{0,1\}^2$, that is $x_2$, is randomized by $W^{(b)}$). 
\item If $ c \in \mathcal{C}$, the input of the corresponding $\mathcal{W}^{(c)}$ is set to $\ket{0}$, the eigenstate of the Pauli $Z$ operator with eigenvalue 1 (since $X$ part of the input $x_1x_2 \in \{0,1\}^2$, that is $x_1$, is completely randomized by $W^{(c)}$). 

\item If  $ d \in \bar{\mathcal{D}}$, the input of the corresponding $\mathcal{W}^{(d)}$ is set to half of an EPR pair. The other half of the EPR pair is given to the decoder.
\end{itemize}

 With a slight abuse of notation, we shall denote $\mathcal{A}$, $\mathcal{B}$, $\mathcal{C}$ and $\bar{\mathcal{D}}$ as qudit quantum systems with dimensions $2^{|\mathcal{A}|}$, $2^{|\mathcal{B}|}$, $2^{|\mathcal{C}|}$ and $2^{|\bar{\mathcal{D}}|}$, respectively. Let a quantum state $\rho_{\cal A}$ on the system ${\cal A}$ is encoded by supplying it as input to the virtual channels corresponding to $a \in \mathcal{A}$. Define a maximally entangled state $\Phi_\mathcal{DD'}$ as follows
\begin{equation}\label{eq:max_entangled_state}
\Phi_{\cal \bar{D}\bar{D'}} = \otimes_{d\in \bar{\mathcal{D}}} \Phi_{dd'},
\end{equation}
where indices $d$ and $d'$ indicate the $d$-th qubits of  systems $\bar{\mathcal{D}}$ and $\bar{\mathcal{D}}'$, respectively, and $\Phi_{dd'}$ is the density matrix corresponding to an EPR pair.  Define $\rho_\mathcal{B}^+ \eqdef \otimes_{b \in \mathcal{B}} \ket{+}\bra{+}_b$ and $\rho_\mathcal{C}^0 \eqdef \otimes_{c \in \mathcal{C}} \ket{0}\bra{0}_c$. Let also $G_q$ denote the quantum polar transform, that is, the $N$-qubit Clifford unitary obtained by applying the two-qubit Clifford unitary $L$ for $n$ levels of recursion, as depicted in Figure~\ref{fig:quantum_pol_constr}. 

\medskip The encoded state, denoted by $\varphi_{{\cal ABC\bar{D}\bar{D'}}}$, is obtained by applying $G_q\otimes I_{{\cal \bar{D'}}}$ on the system $\mathcal{ABC\bar{D}\bar{D'}}$ as follows
%
%
\begin{equation}
\varphi_{{\cal ABC\bar{D}\bar{D'}}}  \eqdef (G_q\otimes I_{{\cal \bar{D'}}}) (\rho_{\cal A} \otimes \rho_\mathcal{B}^+  \otimes  \rho_\mathcal{C}^0  \otimes \Phi_{\cal \bar{D}\bar{D'}}) (G_q^\dagger \otimes I_{{\cal \bar{D'}}}).
\end{equation}
As no errors occur on the system $\bar{\mathcal{D}}'$, the following is the channel output,
\begin{equation}
\psi_{{\cal ABC\bar{D}\bar{D'}}} \eqdef ({\cal W}^{\otimes N} \otimes I_{{\cal \bar{D'}}}) (\varphi_{{\cal ABC\bar{D}\bar{D'}}}).
\end{equation}
Since ${\cal W}$ is a Pauli channel, we have that
\begin{equation}
\psi_{{\cal ABC\bar{D}\bar{D'}}}  = (E_{{\cal ABC\bar{D}}}G_q\otimes I_{{\cal \bar{D'}}}) (\rho_{\cal A} \otimes \rho_\mathcal{B}^+  \otimes  \rho_\mathcal{C}^0  \otimes \Phi_{\cal \bar{D}\bar{D'}}) (G_q^\dagger E_{{\cal ABC\bar{D}}}^\dagger \otimes I_{{\cal \bar{D'}}}),
\end{equation}
for some $N$-qubit Pauli error $E_{{\cal ABC\bar{D}}} \in \bar{P}_N$.

\subsection{Decoding} 

The decoding is similar to~\cite{DGMS19}, and which is performed in the three steps given below.

\medskip \noindent {\bf Step 1: Apply the inverse quantum polar transform on the channel output state.}  Applying $G_q^\dagger$ on the output state $\psi_{{\cal ABC\bar{D}\bar{D'}}}$, we have that
\begin{align}
G_q^\dagger \psi_{{\cal ABC\bar{D}\bar{D'}}} G_q  &= (G_q^\dagger E_{{\cal ABC\bar{D}}}G_q\otimes I_{{\cal \bar{D'}}}) (\rho_{\cal A} \otimes \rho_\mathcal{B}^+  \otimes  \rho_\mathcal{C}^0  \otimes \Phi_{\cal \bar{D}\bar{D'}}) (G_q^\dagger E_{{\cal ABC\bar{D}}}^\dagger G_q \otimes I_{{\cal \bar{D'}}}) \nonumber \\
& =( E_{{\cal ABC\bar{D}}}' \otimes I_{{\cal \bar{D'}}}) (\rho_{\cal A} \otimes \rho_\mathcal{B}^+  \otimes  \rho_\mathcal{C}^0  \otimes \Phi_{\cal \bar{D}\bar{D'}}) (E_{{\cal ABC\bar{D}}}' \otimes I_{{\cal \bar{D'}}}) \nonumber,
\end{align}
where $E_\mathcal{ABCD}' \eqdef G_q^\dagger E_{{\cal ABC\bar{D}}}G_q$. Since $G_q$ is a $N$-qubit Clifford unitary, it follows that $ E_{{\cal ABC\bar{D}}}' \in \bar{P}_N$ is also a Pauli error.

\medskip \noindent \textbf{Step 2: Quantum measurement.} Let $E_\mathcal{ABC\bar{D'}}' = \otimes_{a \in \mathcal{A}} E_a' \otimes_{b \in \mathcal{B}} E_b' \otimes_{c \in \mathcal{C}} E_c' \otimes_{d \in \bar{\mathcal{D}}} E_d'$, where $E_a', E_b', E_c', E_d' \in \bar{P}_1$. We know that any $E_i' \in \bar{P}_1$ can be written as $X^{u_1} Z^{u_2}$, where $u_1 u_2 \in \{0, 1\}^2$. The decoder performs the Pauli $X$ measurement on each $b \in \mathcal{B}$, which determines the $Z$ part ($u_2$) corresponding to $E_{b}'$, and the Pauli $Z$ measurement on each $c \in \mathcal{C}$, which determines the $X$ part ($u_1$) corresponding to $E_{c}'$. Finally, the decoder performs the \emph{Bell measurement}, that is, the measurement corresponding to the Pauli operators $X \otimes X$ and $Z \otimes Z$, on the two-qubit system $dd'$ for each $d \in \bar{\mathcal{D}}$, which determines both $X$ and $Z$ parts ($u_1u_2$) corresponding to $E_d'$.
\medskip

\medskip \noindent {\bf Step 3: Decode the classical counterpart polar code.} Note that when the {\em all-identity vector}  $I^N \in \bar{P}_1^N$ is input to the $N$ instances of the classical counterpart ${\cal W}^\#$, denoted by ${{\cal W}^\#}^{N}$, the error $E_{{\cal ABC\bar{D}}} \in \bar{P}_N$ can be considered as an output of ${{\cal W}^\#}^{N}$. As $\mathcal{W}^\#$ is a symmetric channel, we have that ${{\cal W}^\#}^{N}(E_{{\cal ABC\bar{D}}} \mid I^N) = {{\cal W}^\#}^{N}(I^N \mid E_{{\cal ABC\bar{D}}})$, therefore, we can equivalently consider $I^N$ as the observed channel output, and  $E_{{\cal ABC\bar{D}}}$ (unknown) the channel input. Hence, we have been given, 
\begin{itemize}

\item $u_2$ corresponding to $E'_b$ for any $b \in \mathcal{B}$.

\item $u_1$ corresponding to $E'_c$ for any $c \in \mathcal{C}$.

\item  $u_1 u_2$ corresponding to $E'_d$ for any $d \in \bar{\mathcal{D}}$.

\item A noisy observation (namely $I^N$) of the error $E_{{\cal ABC\bar{D}}} = G_c E'_{{\cal ABC\bar{D}}}$, where  $ G_c E'_{{\cal ABC\bar{D}}} \break = G_q E'_{{\cal ABC\bar{D}}} G_q^\dagger$.

\end{itemize}

Based on the above, we can use classical polar decoding, namely the \emph{successive cancellation} decoding, to recover the value of $u_1u_2$ corresponding to $E'_a$ for all $ a \in \mathcal{A}$, $u_1$ corresponding to $E'_b$ for all $ b \in \mathcal{B}$, and $u_2$ corresponding to $E'_c$ for all $ c \in \mathcal{C}$. 


\subsection{Number of Preshared EPR pairs}

In this section, we give an upper bound on $\frac{|\mathcal{D}|}{N}$, that is, the fraction of virtual channels requiring preshared EPR pairs, and also a lower bound on $\frac{|\mathcal{B}| + |\mathcal{C}|}{N}$, that is, the fraction of virtual channels frozen in either the Pauli $X$ or $Z$ basis. 

\begin{proposition} \label{lem:upper_D_lower_B+C}
Following inequalities hold for sufficiently large $N$,
\begin{itemize}
\item[(a)] $\frac{|\mathcal{D}|}{N} \leq  Z(W^{[1]}) Z(W^{[2]}).$
\item[(b)] $\frac{|\mathcal{B}| + |\mathcal{C}|}{N} \geq 2 - I(W) - 2 Z(W^{[1]}) Z(W^{[2]}) ,$ where $I(W)$ is the symmetric mutual information of $W$.
\end{itemize}

\end{proposition}

\begin{proof}
\textbf{Point $(a)$:} From~(\ref{eq::ineq_1})-(\ref{eq::ineq_4}), we have the following for any $W$,
\begin{align}
Z({W^0}^{[1]})Z({W^0}^{[2]}) &\leq \left(2 - Z(W^{[2]})\right) Z(W^{[1]})Z(W^{[2]}). \nonumber \\
Z({W^1}^{[1]})Z({W^1}^{[2]}) &\leq Z_1(W)Z(W^{[1]})Z(W^{[2]}). \nonumber
\end{align}
Using the above two equations, we have that
\begin{small}
\begin{align}
\sum_{i_1 \in \{0, 1\}} Z({W^{i_1}}^{[1]})Z({W^{i_1}}^{[2]}) &\leq 2 Z(W^{[1]})Z(W^{[2]}) - \left(Z(W^{[2]}) - Z_1(W)\right)Z(W^{[1]})Z(W^{[2]}), \label{eq:z1z2_supmar0} \\
&\leq 2 Z(W^{[1]})Z(W^{[2]}),\label{eq:z1z2_supmar}
\end{align}\end{small}
\noindent where the second inequality follows from $Z(W^{[2]}) \geq Z_1(W)$. Applying~(\ref{eq:z1z2_supmar}) recursively, for any $W^{(i)}$, with $i_1 \cdots i_n \in \{0, 1\}^n$ being the binary representation of $i \in \{0, \dots, N - 1\}$, we have that
\begin{align} \label{eq:z1z2_supmar_1}
\sum\limits_{i = 0}^{N - 1} Z({W^{(i)}}^{[1]}) Z({W^{(i)}}^{[2]}) \leq 2^n Z(W^{[1]}) Z(W^{[2]}).
\end{align}
\medskip\noindent We know from Theorem~\ref{thm:main_multi_level} that for sufficiently large $N = 2^n$, any  $ i \in \{0, \dots, N-1\}$ belongs to one of the sets $\mathcal{A}$, $\mathcal{B}$, $\mathcal{C}$ and $\mathcal{D}$ with probability 1. Further, we have that
\[
     Z({W^{(i)}}^{[1]}) Z({W^{(i)}}^{[2]}) \to 
\begin{cases}
    1,& \text{if }  i  \in \mathcal{D}. \\
    0,              & \text{otherwise}.
\end{cases}
\]
Therefore, from~(\ref{eq:z1z2_supmar_1}), it follows that
\begin{align*}
|\mathcal{D}| \leq N Z(W^{[1]}) Z(W^{[2]}).
\end{align*}

\noindent \textbf{Point $(b)$:} Recursively applying~(\ref{eq:I_preserve}), we have that
\begin{equation}
\sum_{i=0}^{N-1} I(W^{(i)}) = N I(W). 
\end{equation}
We know from Section~\ref{sec:noisy_half_noiseless} that $I(W^{(i)}) \to 2$ for $i \in \mathcal{A}$, $I(W^{(i)}) \to 1$ for $i \in \mathcal{B} \cup \mathcal{C}$, and $I(W^{(i)}) \to 0$ for $i \in \mathcal{D}$. Thus, we have that
\begin{align}
2|\mathcal{A}| + |\mathcal{B}| + |\mathcal{C}| = N I(W).
\end{align}
Any $i$ belongs to one of the sets $\mathcal{A}$, $\mathcal{B}$, $\mathcal{C}$ and $\mathcal{D}$ with probability $1$, therefore, 
\begin{align}
\frac{|\mathcal{A}| + |\mathcal{B}| + |\mathcal{C}|  + |\mathcal{D}|}{N} \to 1.
\end{align}
From the above two equations, we have that
\begin{align}
|\mathcal{B}| + |\mathcal{C}| + 2|\mathcal{D}| \approx N \left( 2 - I(W) \right).
\end{align}
Since $|\mathcal{D}| \leq N Z(W^{[1]}) Z(W^{[2]})$ from part $(a)$, we have that
\begin{equation*}
|\mathcal{B}| + |\mathcal{C}|  \geq N \left( 2 -  I(W) - 2 Z(W^{[1]}) Z(W^{[2]})  \right). \qedhere
\end{equation*}
\end{proof}

\noindent The upper bounds in points $(a)$ and $(b)$ of the above lemma are not strict in general as one can get a stronger bound by recursively applying~(\ref{eq:z1z2_supmar0}) instead of~(\ref{eq:z1z2_supmar}) to evaluate $\sum_{i = 0}^{N - 1} Z({W^{(i)}}^{[1]}) Z_2({W^{(i)}}^{[2]})$ in~(\ref{eq:z1z2_supmar_1}). However, here it is not possible to apply~(\ref{eq:z1z2_supmar0}) recursively as we only have upper bound for $Z({W^{i_1}}^{[2]})$ when $ i_1 = 1$.


\subsection{Speed of Polarization}

The reliability of the successive cancellation decoding depends on the speed of polarization, that is, if polarization happens fast enough, the block error probability of the successive cancellation decoding goes to zero. In this section, using the results from~\cite{at09}, we give a fast polarization property, which ensures reliable decoding of the quantum polar code constructed in the previous section.

\begin{proposition}
Let $W:=\mathcal{W}^\#$ be the classical counterpart of a CMP channel $\mathcal{W}$, and consider the quantum polar construction on $\mathcal{W}$ for $n$ polarization steps, using the two-qubit Clifford gate $L$ as channel combining operation. If $P_e^B$ is the block error probability of the successive cancellation decoding, then we have the following as $n \to \infty$,
\begin{equation}
P_e^B = \mathcal{O}(2^{n} 2^{-2^{\beta n}}),
\end{equation}
for any  $0 < \beta <  \frac{1}{4}$.
\end{proposition}

%
%

\begin{proof}
From~(\ref{eq:bad_bad_can1})-(\ref{eq:good_good_can1}) and~(\ref{eq:bad_bad_can2})-(\ref{eq:good_good_can2}), for all $d=1, 2$ and $\omega_i \in \{0, 1\}^2$, we have that

\[ Z({W^{\omega_i}}^{[d]})  \leq 
\begin{cases}
2 Z(W^{[d]}), \text{ with probability } \frac{1}{2}\\
Z(W^{[d]})^2, \text{ with probability } \frac{1}{2}
\end{cases} 
\]
Therefore, from~\cite{at09}, for any sequence $ \omega = \omega_1 \cdots \omega_m$, with $n =2m$, and $\omega_k \in \{0, 1\}^2, \forall k > 0$, such that  $Z({W^{(\omega)}}^{[d]}) \to 0$ as $m \to \infty$, we have that
\begin{equation} \label{eq:arikan_fast_polarization}
Z({W^{(i)}}^{[d]}) \leq 2^{-2^{\alpha m}}, \text{ for any } 0< \alpha < \frac{1}{2}.
\end{equation}
From~(\ref{eq:event_S^d}), the condition $Z({W^{(\omega)}}^{[d]}) \to 0$ as $m \to \infty$ implies that  $\omega \in S^{[d]}(\delta)$ with $\delta \to 0 $. From~(\ref{eq:pol_1})-(\ref{eq:pol_4}), we know that $S^{[1]}(\delta) = A \cup B$ and $S^{[2]}(\delta) = A \cup C$. Therefore, the above equation holds for $\omega \in A \cup B$, when $d =1$, and $\omega \in A \cup C$, when $d = 2$.

\medskip\noindent From~\cite[Proposition 2]{sta09}, the symbol error probability of the maximum likelihood decoder, denoted by $P_e$, is upper bounded as $P_e(W) \leq 3 Z(W)$, and $P_e(W^{[d]}) \leq Z(W^{[d]})$. Therefore, the block error probability of the successive cancellation decoding, $P_e^B$, can be upper bounded for sufficiently large codelength $2^n$ as follows
\begin{align}
P_e^B &\leq \sum_{a \in \mathcal{A}}3 Z(W^{(a)}) + \sum_{b \in \mathcal{B}} Z({W^{(b)}}^{[1]}) + \sum_{c \in \mathcal{C}} Z({W^{(c)}}^{[2]}) \nonumber  \\
& \leq  \sum_{a} 2 \left[Z({W^{(a)}}^{[1]}) + Z({W^{(a)}}^{[2]}) \right] + \sum_{b \in B} Z({W^{(b)}}^{[1]}) + \sum_{c \in C}Z({W^{(c)}}^{[2]})  \nonumber \\
&\leq ( 	4|\mathcal{A}| +|B| + |C|) 2^{-2^{\alpha m}}, \text{ as }  m \to \infty  \nonumber \\
&\leq 2^{n+2} 2^{-2^{\alpha m}} \nonumber,
\end{align}
where the second inequality uses [Lemma~\ref{lem:inequal}, Point 2] and the third inequality follows from~(\ref{eq:arikan_fast_polarization}). Therefore, $P_e^B = \mathcal{O}(2^{n} 2^{-2^{\beta n}})$ for any $0 < \beta <  \frac{1}{4}$.
\end{proof}
The above proposition implies that $P_e^B \to 0$ as $N \to \infty$, hence, the decoding is reliable for sufficiently large $N$.

\section{An alternative construction} \label{sec:alt_constr}
In this section, we introduce an alternative construction, the goal of which is to improve the speed of the multilevel polarization. For a quantum erasure channel, we show in the next section with the help of a computer program that the multilevel polarization occurs significantly faster for the alternative construction compared to the previous construction.

\medskip Firstly, we note the following circuit equivalence,
\begin{figure}[H]
   \centering
   \resizebox{ 0.9 \linewidth }{!}{
     \begin{tikzpicture}
     \begin{scope}
        \draw
            (0,0) node[draw](hd1){$H$} 
            (hd1.east) ++(1, 0) node[draw](can1) {$\mathcal{W}$} 
            (hd1.east) to (can1.west)
            (can1.east) to ++(1, 0)
            
            (hd1) ++(0, -1) node[draw] (hd2) {$H$} 
            (hd2.east) ++(1, 0) node[draw](can2) {$\mathcal{W}$} 
            (hd2.east) to (can2.west)
            (can2.east) to ++(1, 0)
            
            (hd2) ++(0, -2) node[draw] (hd3) {$H$} 
            (hd3.east) ++(1, 0) node[draw](can3) {$\mathcal{W}$} 
            (hd3.east) to (can3.west)
            (can3.east) to ++(1, 0)

            (hd3) ++(0, -1) node[draw] (hd4) {$H$} 
            (hd4.east) ++(1, 0) node[draw](can4) {$\mathcal{W}$} 
            (hd4.east) to (can4.west)
            (can4.east) to ++(1, 0)
            ;
         \draw 
         (hd1) to ++ (-1, 0) node[not](n1){} 
          (hd2) to ++(-1, 0) node[phase](c1){} 
          (n1) to ++(-2.75, 0) ++(-0.5, 0) node[draw](hd6){$H$} 
          (hd6.east) to ++(0.20, 0)
          (hd6.west) to ++(-1, 0) node[not](n2){}  to ++(-1, 0)
          (c1) to ++(-0.75, 0) to ++(-2, -2) ++(-0.5, 0) node[draw](hd5){$H$}
          (hd5.east) to ++(0.20, 0)
          (hd5.west) to ++(-1, 0) node[not](n3){} to ++(-1, 0)
          
         (hd3) to ++ (-1, 0) node[not](n4){} 
          (n4) to ++ (-0.75, 0) to ++(-2, 2) ++(-0.5, 0) node[draw] (hd7){$H$}
          (hd7.east) to ++(0.20, 0)
          (hd7.west) to ++(-1, 0) node[phase](c2){} to ++(-1, 0)
          (hd4) to ++(-1, 0) node[phase](c4){} 
          (c4) to ++(-2.75, 0) ++(-0.5, 0) node[draw] (hd8){$H$}
          (hd8.east) to ++(0.20, 0)
          (hd8.west) to  ++(-1, 0) node[phase](c3){} to ++(-1, 0)
           ;
          \draw
           (n1) to (c1)
           (n2) to (c2)
           (n4) to (c4)
           (n3) to (c3)
           (can1) ++(2, -2.0) node[](){$\equiv$}
           ;
\end{scope}
\begin{scope}[xshift = 10 cm]
  \draw
            (0,0) node[draw]  (can1) {$\mathcal{W}$} 
            (can1.east) to ++(1, 0)
            (can1) ++(0, -1) node[draw] (can2) {$\mathcal{W}$} 
            (can2.east) to ++(1, 0)
            (can1) ++(0, -3) node[draw] (can3) {$\mathcal{W}$}
            (can3.east) to ++(1, 0)
            (can3) ++(0, -1) node[draw] (can4) {$\mathcal{W}$}   
            (can4.east) to ++(1, 0)
            ;
          \draw 
          (can1) to ++ (-1, 0) node[phase](n1){} 
          (can2) to ++(-1, 0) node[not](c1){} 
          (n1) to ++(-3.75, 0) node[not](n2){} to ++(-1, 0)
          (c1) to ++(-0.75, 0) to ++(-2, -2) to ++(-1, 0) node[not](n3){} to ++(-1, 0)
          (can3) to ++ (-1, 0) node[phase](n4){} 
          (n4) to ++ (-0.75, 0) to ++(-2, 2) to ++(-1, 0) node[phase](c2){} to ++(-1, 0)
          (can4) to ++(-1, 0) node[not](c4){} 
          (c4) to ++(-3.75, 0) node[phase](c3){} to ++(-1, 0)
           ;
           \draw
           (n1) to (c1)
           (n2) to (c2)
           (n4) to (c4)
           (n3) to (c3)
           ;
\end{scope}
\end{tikzpicture}
}

{ (a) \hspace*{70mm} (b) }
\caption{(a) and (b) are equivalent quantum circuits.}
\end{figure}
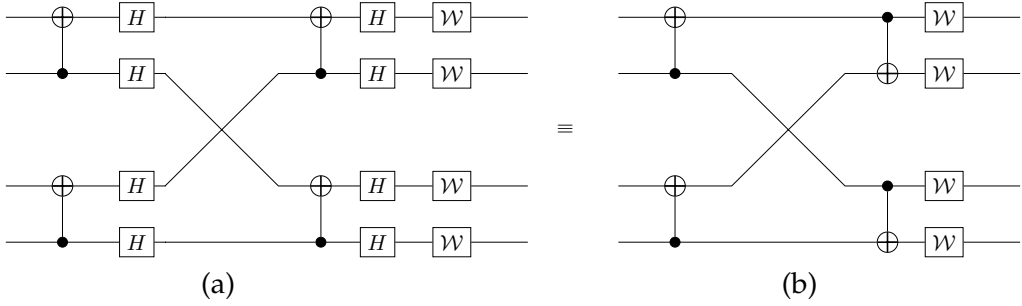
 In circuit $(b)$, the CNOT gate is used in both the first and second polarization step, however, the control and target are interchanged after the first step. To make this clear, we denote by $L_1$ and $L_2$ the CNOT gate in the first and second step, respectively. The quantum circuits $(a)$ and $(b)$ are equivalent in the sense that given any $4$-qubit quantum state as input, the outputs of quantum circuits $(a)$ and $(b)$ are identical. Therefore, the virtual channels obtained after two steps of channel combining and splitting are equal for both circuits $(a)$ and $(b)$. Hence, the multilevel polarization theorem from the previous section also holds when CNOT gates $L_1$ and $L_2$ are used as channel combining operation alternatively for odd and even  polarization steps, respectively. In other words, $L_1$ is used to combine two copies of $\mathcal{W}$, and then $L_2$ is used to combine two copies of $\mathcal{W}^{i_1}$, for all $i_1 \in \{0, 1\}$, again $L_1$ is used to combine two copies of $\mathcal{W}^{i_1i_2}$, for all $i_1, i_2 \in \{0, 1\}$, and so on.
 
\medskip\noindent Here, we propose an alternative construction, where instead of using $L_1$ and $L_2$ for odd and even steps of polarization, an optimal choice is made at each polarization step, using the classical counterpart viewpoint as follows.

\medskip   Let $\Gamma_1$ and $\Gamma_2$ be the permutations asscoiated with $L_1$ and $L_2$, respectively. We define $ T(\Gamma, W) \eqdef Z\left((W \varoast_\Gamma W)^{[1]}\right) + Z\left((W \varoast_{\Gamma} W)^{[2]}\right)$, where $W$ is the classical counterpart of $\mathcal{W}$. For combining two copies of $W$, the permutation $\Gamma \in \{\Gamma_1, \Gamma_2\}$ is selected as channel combining operation if the following holds,
\begin{equation} \label{eq:choose_channel_comb}
T(\Gamma, W) = \displaystyle \text{min}\Big\{T(\Gamma_1, W), T(\Gamma_2, W)\Big\}.
\end{equation}
 A similar selection process takes place at each polarization step, so that two copies of a virtual channel $W^{i_1 \cdots i_k}$ are combined using the permutation $\Gamma^{(i_1 \cdots i_k)} \in \{\Gamma_1, \Gamma_2\}$ minimizing $T(\Gamma^{(i_1 \cdots i_k)},  W^{i_1 \cdots i_k} )$.


\medskip\noindent We now give the following lemma for the Bhattacharya parameter of partial channels associated with virtual channels, $W \boxast W$ and $W \varoast W$, using permutations $\Gamma_1$ and $\Gamma_2$.
\begin{lemma}
Let $ W^0 \eqdef W \boxast W$ and $ W^1 \eqdef W \varoast W$, and for $x \leq a$, $y \leq b$, we denote $(x, y) \leq (a, b)$. When $\Gamma_1$ is used as channel combining operation, we have that
 \[
     \left(Z({W^i}^{[1]}), Z({W^i}^{[2]})\right)  \leq
\begin{cases}
   \left(Z(W^{[1]}), 2Z(W^{[2]}) - Z(W^{[2]})^2 \right) ,& \text{when }  i = 0,\\
    \left(Z(W^{[1]}), Z_1(W) Z(W^{[2]}) \right),              &  \text{when }  i = 1,
\end{cases}
\]
and when $\Gamma_2$ is used as channel combining operation, we have that
\[
     \left(Z({W^i}^{[1]}), Z({W^i}^{[2]})\right) \leq 
\begin{cases}
   \left(2Z(W^{[1]}) - Z(W^{[1]})^2, Z(W^{[2]})  \right) ,& \text{when }  i = 0.\\
    \left(Z_2(W) Z(W^{[1]}), Z(W^{[2]}) \right),              &  \text{when }  i = 1.
\end{cases}
\]
\end{lemma}
\begin{proof}
We have omitted the proof of the lemma as it is basically the same proof as in  Lemma~\ref{lem:partial_canvirtual_inequality}.
\end{proof}
 It can be verified from the above inequalities that~(\ref{eq:z1z2_supmar}), \emph{i.e.}, {\small $\sum_{i_1 \in \{0, 1\}} Z({W^{i_1}}^{[1]})Z({W^{i_1}}^{[2]}) \break \leq 2 Z(W^{[1]})Z(W^{[2]})$}, holds for both $\Gamma_1$ and $\Gamma_2$. This implies that the upper bound on the number of preshared EPR pairs from point $(a)$ of Proposition~\ref{lem:upper_D_lower_B+C} holds for the alternative construction. It is also easy to verify that point $(b)$ of Proposition~\ref{lem:upper_D_lower_B+C} holds as well.





\section{Quantum Erasure Channel} \label{sec:quantm_erasure}

In this section, using both the first and second construction, we construct quantum polar codes for a quantum erasure channel with the help of a computer program, and compare the two constructions in terms of their speeds of polarization.

\medskip\noindent Consider the following quantum erasure channel with erasure probability $\epsilon > 0$, 
\begin{align}
\mathcal{W}_E(\rho_A) = (1 - \epsilon) \ket{0}\bra{0}_F \otimes \rho_A  + \epsilon \ket{1} \bra{1}_F \otimes \frac{I_A}{2}.
\end{align}
The receiver is given a classical flag $F$ together with the quantum output $A$. If $F$ is found to be in the state $\ket{0}$, the output state is equal to the input state $\rho_A$, and if it is in the state $\ket{1}$, the output state is the maximally mixed state $\frac{I_A}{2}$. The quantum erasure channel is a CMP channel as it can be written as,
\begin{align}
\mathcal{W}_E(\rho_A)  = (1 - \epsilon) \ket{0}\bra{0}_F \otimes \mathcal{W}_0(\rho_A) + \epsilon \ket{1} \bra{1}_F \otimes \mathcal{W}_1(\rho_A),
\end{align}
where $\mathcal{W}_0(\rho_A) = \rho_A$ and $\mathcal{W}_1(\rho_A) = \frac{1}{4} [ \rho_A + X \rho_A X + Y \rho_A Y + Z \rho_A Z ] = \frac{I_A}{2}$ for any $\rho_{A}$, are clearly Pauli channels. 
Therefore, the classical counterpart channel $\mathcal{W}_E^\#$ is the classical mixture of Pauli channels $ \mathcal{W}_0^\#$ and $ \mathcal{W}_1^\#$ with probabilities $1 - \epsilon$ and $\epsilon$, respectively. Here, $ \mathcal{W}_0^\#$ is the identity channel as $\mathcal{W}_0^\#(i \mid j) = \delta_{ij}, \forall i, j \in \{0, 1\}^2$, and $ \mathcal{W}_1^\#$ completely randomizes the two-bit input as $\mathcal{W}_1^\#( i \mid j) = \frac{1}{4}, \forall i, j \in \{0, 1\}^2$. Thus, $\mathcal{W}_E^\#$ can be considered as a classical erasure channel with two-bits $x_1 x_2 \in \{0, 1\}^2$ as input and the erasure probability $\mathcal{W}_E^\#(? , ? \mid x_1 , x_2) = \epsilon$. Here, symbol $?$ represents the erasure of a bit. 

\medskip\noindent For the sake of clarity, we denote $ W \eqdef \mathcal{W}_E^\#$ from now on. For $W$, the two bits of the input $x_1,x_2$ is either transmitted perfectly with the probability $ 1 - \epsilon$, or both bits are erased with the probability $\epsilon$. However, polarizing $W$ yields virtual channels that may erase only one bit either $x_1$ or $x_2$ (see also Lemma~\ref{lem:erasure_bht} below). For this reason, we define a more general erasure channel $W'$, referred to as the bit-level erasure channel, as follows.
\begin{definition}[Bit-level erasure channel]
A bit-level erasure channel is defined by the following transition probabilities,
\begin{equation*}
W'(?,x_2|x_1,x_2)=\epsilon_1,  W'(x_1,?|x_1,x_2) = \epsilon_2, W'(?,?|x_1,x_2) = \epsilon_3, \forall x_1, x_2 \in \{0, 1\}. 
\end{equation*}
\end{definition}
The erasure channel $W$ is a special case of the bit-level erasure channel with $\epsilon_1 = \epsilon_2 = 0$, and $ \epsilon_3 = \epsilon$. In the next lemma, we give $Z(W'^{[1]}), Z(W'^{[2]}), Z_1(W')$ and $Z_2(W')$.

\begin{lemma} \label{lem:erasure_err}
The following equalities hold for a  bit-level erasure channel $W'$,
 \begin{align*}
Z(W'^{[1]}) &=  Z_2(W') =  \epsilon_1 + \epsilon_3.  \\
   Z(W'^{[2]}) &=  Z_1(W') = \epsilon_2 + \epsilon_3.
\end{align*} 
\end{lemma}
\begin{proof}
Given $x_1x_2 \in \{0, 1\}^2$ as input to $W'$, the bits $x_1$ and $x_2$ are inputs to the partial channels $W'^{[1]}$ and $W'^{[2]}$, respectively. It is not very difficult to see that $W'^{[1]}$ and $W'^{[2]}$ are binary-input erasure channels with erasure probabilities $\epsilon_1 + \epsilon_3$ and $\epsilon_2 + \epsilon_3$, respectively. Since the Bhattacharyya parameter is equal to the erasure probability for a binary-input erasure channel, it follows that $Z(W'^{[1]}) = \epsilon_1 + \epsilon_3$ and $Z(W'^{[2]}) = \epsilon_2 + \epsilon_3$.

\medskip\noindent Moreover, $Z_1(W') = \epsilon_2 + \epsilon_3$ as for any $x_1, x_2$, $\sqrt{W(y|x_1,x_2)W(y|x_1,x_2 \oplus 1)}$ is non-zero only when $y = x_1, ? \text{ or } y = ?,?$. Similarly, $Z_2(W') = \epsilon_1 + \epsilon_3$.
\end{proof}
\noindent Taking advantage of the above lemma, we will only use quantities $Z(W'^{[1]})$ and $Z(W'^{[2]})$ from now on. Also, from~(\ref{eq:sym_mut_info}) and Lemma~\ref{lem:erasure_err}, the symmetric mutual information of $W'$ is given by, 
\begin{equation} \label{eq:information_erasure}
I(W') = 2 -Z(W'^{[1]}) - Z(W'^{[2]}).
\end{equation}




\subsection{First construction}
Here, we consider the quantum polar code construction given in Section~\ref{sec:multilevel}. Firstly, we prove the following lemma for the partial channels. 
\begin{lemma} \label{lem:erasure_bht}
Given $W'$ is a bit-level erasure channel, let $W'^{0} \eqdef W' \boxast_{\Gamma} W'$ and $W'^{1} \eqdef W' \varoast_{\Gamma} W' $ be the synthesized virtual channels for the channel combining operation $\Gamma = \Gamma(L)$ (Figure~\ref{fig:perm_comb_2}). Then, $W'^{0}$ and $W'^{1}$ are also bit-level erasure channels and the inequalities for partial channels in~(\ref{eq::ineq_1})-(\ref{eq::ineq_4}) are equalities, that is,
\begin{align}
Z({W'^{0}}^{[1]}) &= 2Z(W'^{[2]}) - Z(W'^{[2]})^2. \\
Z({W'^{0}}^{[2]}) &= Z({W'^{0}}^{[1]}). \\
Z({W'^{1}}^{[1]}) &= Z(W'^{[2]})^2.  \\
Z({W'^{1}}^{[2]}) &= Z(W'^{[1]}).
\end{align}

\end{lemma}

\begin{proof} 
The erasure probabilities for $W'^0$, 
\begin{align*}
& \epsilon_1^0 \eqdef W'^{0}(?,x_2|x_1,x_2) = \epsilon_2 + (1 - \epsilon_1 - \epsilon_2 - \epsilon_3) \times (\epsilon_2 + \epsilon_3). \\
& \epsilon_2^0 \eqdef W'^{0}(x_1,?|x_1,x_2) = \epsilon_1 \times (1 - \epsilon_2 - \epsilon_3). \\
& \epsilon_3^0 \eqdef W'^{0}(?,?|x_1,x_2) = \epsilon_3 + \epsilon_1 \times (\epsilon_2 + \epsilon_3).
\end{align*}

\medskip\noindent The erasure probabilities for $W'^1$, 
\begin{align*}
& \epsilon_1^1 \eqdef W'^{1}(?,x_2|x_1,x_2) = \epsilon_2 \times (\epsilon_2 + \epsilon_3).  \\
& \epsilon_2^1 \eqdef W'^{1}(x_1,?|x_1,x_2) = \epsilon_1 + \epsilon_3 \times (1 - \epsilon_2 - \epsilon_3). \\
& \epsilon_3^1 \eqdef W'^{0}(?,?|x_1,x_2) = \epsilon_3 \times (\epsilon_2 + \epsilon_3).
\end{align*}

\medskip\noindent Note that even when $W'$ is an erasure channel, that is, $\epsilon_1 = \epsilon_2 = 0$, we have that $\epsilon_1^0 = \epsilon_2^1 = (1 - \epsilon_3) \epsilon_3 $, which is non-zero except when $\epsilon_3 \in \{0, 1\}$. Therefore, the virtual channels $W'^0$ and $W'^1$ are bit-level erasure channels in general. From Lemma~\ref{lem:erasure_err}, we have that
\begin{align*}
Z({W'^{0}}^{[1]}) &= \epsilon_1^0 + \epsilon_3^0 = 2Z(W'^{[2]}) - Z(W'^{[2]})^2. \\
Z({W'^{0}}^{[2]}) &= \epsilon_2^0 + \epsilon_3^0 = Z(W'^{[1]}). \\
Z({W'^{1}}^{[1]}) &= \epsilon_1^1 + \epsilon_3^1 = Z(W'^{[2]})^2. \\
Z({W'^{1}}^{[2]}) &= \epsilon_2^1 + \epsilon_3^1 = Z(W'^{[1]}).
\end{align*}
\end{proof}
\noindent  Applying Lemma~\ref{lem:erasure_bht} recursively, we may compute $({Z(W^{(i)}}^{[1]}), {Z(W^{(i)}}^{[2]}))$ for any virtual channel $W^{(i)}$.

%


\subsection{Second construction}

Here, we consider the alternative construction proposed in the Section~\ref{sec:alt_constr}. First of all, we give the following Lemma for $\Gamma_1$ and $\Gamma_2$, the permutations associated with the CNOT gates $L_1$ and $L_2$, respectively.

\begin{lemma} \label{lem:erasure_bht_1}
Given a bit-level erasure channel $W'$, let $W'^0 \eqdef W' \boxast_{\Gamma} W'$ and $W'^1 \eqdef W' \varoast_{\Gamma} W'$. Then, for $\Gamma = \Gamma_1$ as channel combining operation, we have that

\medskip\noindent 
\[
     \left(Z({W'^i}^{[1]}), Z({W'^i}^{[2]})\right)  = 
\begin{cases}
   \left(Z(W'^{[1]}), 2Z(W'^{[2]}) - Z(W'^{[2]})^2 \right) ,& \text{when }  i = 0,\\
    \left(Z(W'^{[1]}), Z(W'^{[2]})^2 \right),              &  \text{when }  i = 1,
\end{cases}
\]
 and for $\Gamma = \Gamma_2$ as channel combining operation, we have that
\[
     \left(Z({W'^i}^{[1]}), Z({W'^i}^{[2]})\right)  = 
\begin{cases}
   \left(2Z(W'^{[1]}) - Z(W'^{[1]})^2, Z(W'^{[2]})  \right) , & \text{when }  i = 0.\\
    \left(Z(W'^{[1]})^2, Z(W'^{[2]}) \right),              &  \text{when }  i = 1.
\end{cases}
\]

\begin{proof}
The proof has been omitted as it is basically the same proof as in Lemma~\ref{lem:erasure_bht}.
\end{proof}

\end{lemma}

\noindent Recall from Section~\ref{sec:alt_constr} that $\Gamma \in \{\Gamma_1, \Gamma_2\}$ is chosen as channel combining operation if it satisfies~(\ref{eq:choose_channel_comb}). From Lemma~\ref{lem:erasure_bht_1}, for a virtual channel $W'^{i_1 \cdots i_k}$, we have that 
\begin{align}
T(\Gamma_1, W'^{i_1 \cdots i_k}) &= Z({W'^{i_1 \cdots i_{k}}}^{[1]}) + Z({W'^{i_1 \cdots i_{k}}}^{[2]})^2.  \\
T(\Gamma_2, W'^{i_1 \cdots i_{k}}) &= Z({W'^{i_1 \cdots i_{k}}}^{[1]})^2  + Z({W'^{i_1 \cdots i_{k}}}^{[2]}) .
\end{align}
Therefore, for a virtual channel $W^{i_1 \cdots i_k}$, we first determine the optimal permutation from $\{\Gamma_1, \Gamma_2\}$ using the above two equations, and subsequently compute $( Z({W^{i_1 \cdots i_k i_{k+1}}}^{[1]}), \break  Z({W^{i_1 \cdots i_n i_{k+1}}}^{[2]}))$ using Lemma~\ref{lem:erasure_bht_1}.


\subsection{Numerical Results}
It follows from Lemmas~\ref{lem:erasure_bht} and~\ref{lem:erasure_bht_1} that for a bit-level erasure channel $W'$,~(\ref{eq:z1z2_supmar}) is an equality for both the first and second construction, \emph{i.e.}, $\sum_{i_1 \in \{0, 1\}} Z({W^{i_1}}^{[1]})Z({W^{i_1}}^{[2]}) = 2Z(W^{[1]})Z(W^{[2]})$. Therefore, the upper bound on $|\mathcal{D}|$ and the lower bound on $|\mathcal{B}| + |\mathcal{C}|$ from Proposition~\ref{lem:upper_D_lower_B+C} are also equalities for both first and second constructions. Hence, as $N \to \infty$, we have that
\begin{align}
\frac{|\mathcal{D}|}{N} &\to Z(W'^{[1]})Z(W'^{[2]}),  \\
\frac{|\mathcal{B}| + |\mathcal{C}|}{N} &\to \left(Z(W'^{[1]}) + Z(W'^{[2]}) - 2 Z(W'^{[1]}) Z(W'^{[2]})\right),  \\
\frac{|\mathcal{A}|}{N} &\to \left( 1 - Z(W'^{[1]}) - Z(W'^{[2]}) + Z(W'^{[1]}) Z(W'^{[2]}) \right),
\end{align}
where the first equation follows from~[Proposition \ref{lem:upper_D_lower_B+C},  point (a)], the second equation follows from~[Proposition \ref{lem:upper_D_lower_B+C},  point (b)] and~(\ref{eq:information_erasure}), and the third equation is obtained by using $ \frac{|\mathcal{A}| + |\mathcal{B}| + |\mathcal{C}| + |\mathcal{D}|}{N} \to 1$.

\medskip\noindent We now consider a quantum erasure channel with erasure probability $W(?, ?|x_1,x_2) = 0.1$. From Lemma~\ref{lem:erasure_err}, $Z(W^{[1]}) = Z(W^{[2]}) = 0.1$. From above three equations, it follows that $\frac{|\mathcal{D}|}{N} \to 0.01$, $\frac{|\mathcal{B}| + |\mathcal{C}| }{N} \to 0.18 $ and $\frac{|\mathcal{A}|}{N} \to 0.81$ as $N \to \infty$. Therefore, we have saved $9\%$ of EPR pairs as compared to~\cite{DGMS19}, and are left with only $1\%$ of preshared EPR pairs. For this erasure channel, we perform a numerical simulation for $n = 20$ steps of polarization for both the first and second construction, and compare their speeds of polarization.

\medskip In  Figure~\ref{fig:optimized}, the parameter $ T^{(i)} \eqdef Z({W^{(i)}}^{[1]}) + Z({W^{(i)}}^{[2]})$ is plotted for both first and second constructions after $n = 20$ polarization steps. 
\begin{figure}[H]
\centering
\includegraphics[scale=0.8]{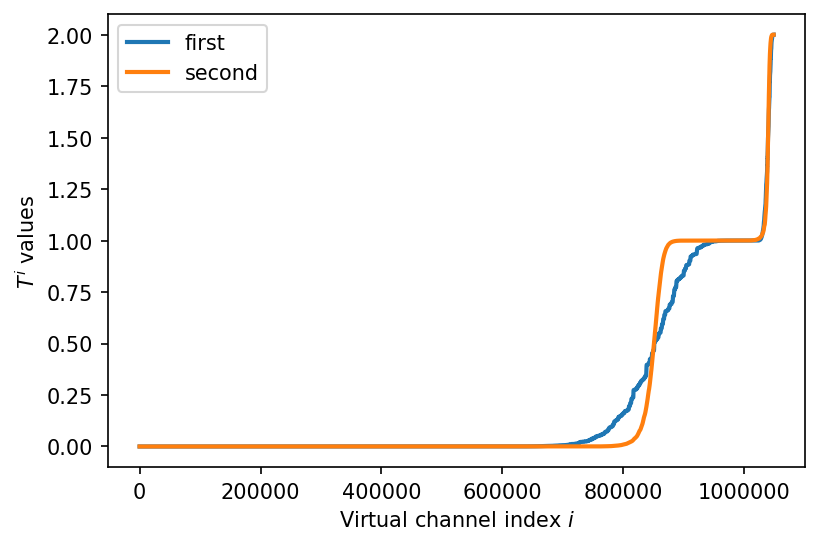}
\caption{$T^{(i)}$ values for a quantum erasure channel with erasure probability $W(?, ? | x_1, x_2) = 0.1$ after $n=20$  polarization steps. The virtual channel indices $i \in \{0,\dots, 2^n-1\}$ are sorted according to increasing $T^{(i)}$ values.  }
\label{fig:optimized}
\end{figure}

The multilevel polarization is evident in the above figure, especially for the second construction, as $T^{(i)}$ approaches faster the limit values $0, 1, \text{ or } 2$. In particular, we have the following,
\begin{itemize}
\item When $T^{(i)} \to 0$, that is, the plateau in the begining of the plot, $ i \in \mathcal{A}$. 

\item When $T^{(i)} \to 1$,  that is, the plateau in the middle of the plot,  $ i \in \mathcal{B} \cup \mathcal{C}$. 

\item When $T^{(i)} \to 2$, that is, the plateau in the end of the plot, $ i \in \mathcal{D}$.
\end{itemize}

\noindent For $\delta= 10^{-6}$, we compare the first and second construction in the following table,
\begin{center}
\begin{tabular}{ | c| c|  c| c| c|}
\hline
& $\frac{|\mathcal{A}|}{N}$& $\frac{|\mathcal{B}| + |\mathcal{C}|}{N}$  &$\frac{|\mathcal{D}|}{N}$ & $\frac{|\mathcal{A}| + |\mathcal{B}| + |\mathcal{C}| + |\mathcal{D}|}{N}$\\ 
 \hline
First construction & $0.49438$ & $0.03021$ & $0.00046$ & $0.52505$ \\ 
 \hline 
second construction &$0.64493$ &$0.07359$ & $0.00071$ & $0.71923$  \\
 \hline
\end{tabular}
\end{center}
Therefore, the numerical simulation suggests that the multilevel polarization happens significantly faster for the second construction compared to the first one.

\section{Conclusion}
For the family of Pauli channels (in fact, more general CMP channels), we have proven multilevel polarization using a fixed two-qubit Clifford unitary as channel combining operation. We have shown that the multilevel polarization can be used to build an efficient quantum code and allows to reduce the number of preshared EPR pairs with respect to~\cite{DGMS19}. Finally, we have presented an alternative construction to improve the speed of polarization and have shown by numerical simulation that the speed of polarization improves significantly for a quantum erasure channel. A natural future direction would be to investigate whether the number of preshared EPR pairs can be further reduced by combining several two-qubit Clifford unitaries.

\section*{Acknowledgments}
This research was supported in part by  the ``Investissements d’avenir'' (ANR-15-IDEX-02) program of the French National Research Agency. AG acknowledges the European Union’s Horizon 2020 research and innovation programme under the Marie Skłodowska-Curie grant agreement No 754303.

\newpage

\appendix

\section{Proof of Lemma~\ref{lem:inequal}}\label{sec:proof-inequality}

\textbf{Point 1:}  The Bhattacharyya parameter of the partial channel $W^{[d_1]}$ is given by,
\begin{align*}
Z(W^{[d_1]}) & = \sum_y \sqrt{W^{[d_1]}(y|0) W^{[d_1]}(y|1)} \nonumber \\
       & = \frac{1}{2}  \sum_y \sqrt{\sum_{\stackindx{l \in \{0, d_1\}} {m \in \{d_2, d_3\}}} W(y|l) W(y|m) }  \nonumber \\
       & \leq \frac{1}{2} \sum_{\stackindx {l \in \{0, d_1\}} {m \in \{d_2, d_3 \}}} \sum_y \sqrt{ W(y|l) W(y|m)} \nonumber \\
       & = Z_{d_2}(W) + Z_{d_3} (W), \nonumber 
\end{align*}
where the second equality follows from~(\ref{eq:def:partial_w_d}),  the third inequality follows from $\sqrt{\sum_x a_x} \leq \sum_x \sqrt{a_x}$, and the fourth equality follows from $ l \oplus m \in \{d_2, d_3\}, \forall l, m$ as $d_3 = d_1 \oplus d_2$ and~(\ref{eq:bhattacharyya_d}).

\medskip\noindent \textbf{Point 2:} For $W^{[d_1]}$, we consider the following two-dimensional vectors:
\begin{align}
 \vec{B}_0(y) & = ( \sqrt{W(y|0)}, \sqrt{W(y|d_1)} ). \nonumber \\
\vec{B}_1(y) & =   (\sqrt{W(y|d_2)}, \sqrt{W(y|d_1 \oplus d_2)}). \nonumber \\
\vec{B}_2(y) & =  (\sqrt{W(y|d_1 \oplus d_2)}, \sqrt{W(y|d_2)}). \nonumber 
\end{align}
Then, we have that
\begin{align*}
 &|\vec{B}_{0}(y)| = \sqrt{ W(y|0) + W(y|d_1)}.  \\
& |\vec{B}_{1}(y)| = |\vec{B}_{2}(y)|  =  \sqrt{W(y| d_2) + W(y| d_1 \oplus d_2)}. \\
&  \vec{B}_{0}(y) \cdot \vec{B}_{1} (y) = \sqrt{W(y|0)} \sqrt{W(y|d_2)} + \sqrt{W(y| d_1 )} \sqrt{W(y| d_1 \oplus d_2)}. \\ 
& \vec{B}_{0}(y) \cdot \vec{B}_{2} (y) = \sqrt{W(y|0)} \sqrt{W(y|d_1 \oplus d_2)} + \sqrt{W(y| d_1 )} \sqrt{W(y| d_2)}.  \\
\end{align*}
From the definitions of  $Z(W^{[i]})$ and $Z_d(W)$, it follows:
\begin{align} \label{eq:sum_y}
Z(W^{[d_1]}) &=  \frac{1}{2}\sum_y |\vec{B}_{0}(y)||\vec{B}_{1}(y)| = \frac{1}{2}\sum_y |\vec{B}_{0}(y)||\vec{B}_{2}(y)|. \\
 Z_{d_2}(W) & = \frac{1}{2} \sum_y \vec{B}_{0}(y) \cdot \vec{B}_{1} (y). \\
Z_{d_3}(W) =  Z_{d_1 \oplus d_2}(W) & = \frac{1}{2} \sum_y \vec{B}_{0}(y) \cdot \vec{B}_{2} (y).
\end{align}
Then, from the Cauchy-Schwartz inequality, we have that

\begin{equation}
Z_d(W) \leq Z(W^{[d_1]}), \text{ for } d = d_2, d_3.
\end{equation}

\medskip\noindent \textbf{Point 3:} $I(W)$ can be written as following~\cite[Lemma 10]{PB12}

{\small
\begin{align}
I(W) = \frac{1}{4} \sum\limits_y \sum_{x \in \bar{P}_1} W(y|x) \text{ log}_2 \frac{W(y|x)}{P(y)} \quad \quad \quad \quad \quad \quad \quad \quad \quad \quad \quad \quad \quad \quad \quad \quad \quad \quad \quad \quad \quad \quad \quad \quad \quad \quad \quad \quad  & \nonumber \\
      = \frac{1}{4} \sum\limits_y \frac{1}{6}\sum\limits_{d\in \{d_1, d_2, d_3\}}  \sum\limits_{x} \big[ W(y|x) \text{ log}_2 \frac{W(y|x)}{P(y)} +  W(y|x \oplus d) \text{ log}_2 \frac{W(y|x \oplus d)}{P(y)} \big] \quad \quad \quad \quad \quad \quad \quad \quad \quad \quad \quad & \nonumber \\
      = \frac{1}{24} \sum\limits_{ d \in \{d_1, d_2, d_3\}} \sum\limits_x W(y|x) \text{ log}_2 \frac{W(y|x)}{\frac{1}{2}[W(y|x) + W(y|x \oplus d)]} +  W(y|x \oplus  d) \text{ log}_2 \frac{W(y|x \oplus d)}{\frac{1}{2}[W(y|x) + W(y|x \oplus d)]} & \nonumber\\
      \quad \quad \quad \quad + \frac{1}{12} \sum\limits_y \sum\limits_{d \in \{d_1, d_2, d_3\}}  \sum\limits_{x} \frac{W(y|x) + W(y|x \oplus d)}{2} \text{ log}_2 \frac{\frac{1}{2}[W(y|x) + W(y|x \oplus d)] }{ P(y)} & \nonumber \\
      = \frac{1}{12} \sum\limits_{d \in \{d_1, d_2, d_3\} } \sum_x I(W_{x, x \oplus d}) + \frac{1}{6} \sum\limits_y  \sum\limits_{d \in \{d_1, d_2, d_3\}} \frac{W(y|0) + W(y|d)}{2} \text{ log}_2 \frac{\frac{1}{2}[W(y|0) + W(y| d)] }{ P(y)} \quad \quad \quad  & \nonumber \\
       \quad \quad \quad \quad \quad \quad  \quad \quad + \frac{1}{12} \sum\limits_y \sum\limits_{d \in \{d_1, d_2, d_3\}} \sum\limits_{x \neq 0, d} \frac{W(y|x) + W(y|x \oplus d)}{2} \text{ log}_2 \frac{\frac{1}{2}[W(y|x) + W(y| x \oplus d)] }{ P(y)} & \nonumber \\
     = \frac{1}{3} \sum\limits_{d \neq 0} I_d(W) + \frac{1}{6} \sum\limits_{d \in \{d_1, d_2, d_3\}} \sum\limits_y \Big[ W^{[d]}(y|0)  \text{ log}_2 \frac{W^{[d]}(y|0) }{ P(y)} + W^{[d]}(y|1)  \text{ log}_2 \frac{W^{[d]}(y|1) }{ P(y)} \Big] \quad \quad \quad \quad \quad \quad & \nonumber \\
    = \frac{1}{3} \sum\limits_{d \neq 0} I_d(W) +  \frac{1}{3} \sum\limits_{d \in \{1, 2, 3\} } I(W^{[d]}), \quad \quad \quad \quad \quad \quad \quad \quad  \quad \quad \quad \quad \quad \quad \quad \quad \quad \quad \quad \quad \quad \quad \quad \quad \quad \quad \quad \quad \quad & \nonumber
\end{align}
}
where $I(W_{x,x'})$ and $I_d(W)$ are defined in Section~\ref{sec:definitions}. Also, $I(W^{[d]})$ is the symmetric mutual information of the binary-input partial channel $W^{[d]}$. Using $I(W_{x, x'}) \leq \sqrt{1-Z(W_{x, x'})^2}$ from~\cite{arikan09}, and concavity of the function $f(x) = \sqrt{1-x^2}$, we have that

\begin{equation} \label{eq:bound_mutual}
 I(W)  \leq  \frac{1}{3} \sum_{d \in \{1, 2, 3\}} \sqrt{1 - Z_d(W)^2} + \frac{1}{3} \sum_{i \in \{1, 2, 3\}} \sqrt{1 - Z(W^{[i]})^2}.
\end{equation}
\hfill

\section{Proof of Lemma~\ref{lem:inequality_bht}} \label{sec:proof_inequality_bht}
\begin{proof}
\textbf{Proof of~(\ref{eq:ineq_Zgood_1}):}
\begin{align}
Z_1(W \varoast W) &= \frac{1}{4} \sum_{ \substack{y_1,y_2, u_1, u_2 \\ v_1, v_2}} \sqrt{(W \varoast W)(y_1,y_2, u_1,u_2|v_1,v_2)(W \varoast W)(y_1,y_2, u_1,u_2|v_1,v_2 + 1)} \nonumber \\
& = \frac{1}{16} \sum_{ \substack{y_1,y_2, u_1, u_2 \\ v_1, v_2}} W(y_1|u_2,u_1+v_1)\sqrt{W(y_2|u_2+v_2,v_1)W(y_2|u_2+v_2+1,v_1)} \nonumber  \\
& = \frac{1}{16} \sum_{ \substack{y_2, u_1, u_2 \\ v_1, v_2}} \sqrt{W(y_2|u_2+v_2,v_1)W(y_2|u_2+v_2+1,v_1)} \nonumber \\
& = \frac{1}{4} \sum_{u_1,u_2} Z_2(W) = Z_2(W).
\end{align}

\noindent \textbf{Proof of~(\ref{eq:ineq_Zgood_2}):}
\begin{align}
Z_2(W \varoast W) &= \frac{1}{4} \sum_{ \substack{y_1,y_2, u_1, u_2 \\ v_1, v_2}} \sqrt{(W \varoast W)(y_1,y_2, u_1,u_2|v_1,v_2)(W \varoast W)(y_1,y_2, u_1,u_2|v_1+1,v_2)} \nonumber \\
& = \frac{1}{16} \sum_{y_1, u_1,u_2} \sqrt{W(y_1|u_2,u_1+v_1)W(y_1|u_2,u_1+v_1+1)} \nonumber \\
& \qquad \qquad \qquad \cdot \sum_{y_2, v_1 , v_2} \sqrt{W(y_2|u_2+v_2,v_1)W(y_2|u_2+v_2,v_1+1)} \nonumber  \\
& = Z_1(W)^2.  
\end{align}
\end{proof}

\section{Proof of Lemma~\ref{lem:partial_canvirtual_inequality}} \label{sec:proof_inequality_partial_bht}

The transition probabilities of the partial channels (see~(\ref{eq:partial_1}) and~(\ref{eq:partial_2})) $(W \varoast W)^{[i]}$ and $(W \boxast W)^{[j]}$ for $i,j \in \{1, 2\}$ is given by

\begin{equation}
(W \boxast W)^{[1]}(y_1, y_2|u_1) = \frac{(W \boxast W)(y_1, y_2|u_1,0) + (W \boxast W)(y_1, y_2|u_1,1)}{2} \quad \quad \quad \quad\quad \quad \quad \quad \quad \quad \quad \quad \nonumber \\
\end{equation}
\begin{align}
& = \frac{1}{8}  \sum_{v_1,v_2} [W(y_1|0 , u_1 + v_1) W(y_2|v_2 , v_1) + W(y_1|1 , u_1 + v_1) W(y_2|v_2+1 , v_1)]   \nonumber \\
&= \frac{1}{8}  \sum_{v_1} [W(y_1|0 , u_1 + v_1) \sum_{v_2} W(y_2|v_2 , v_1) + W(y_1|1 , u_1 + v_1) \sum_{v_2} W(y_2|v_2+1 , v_1)]  \nonumber \\
&= \frac{1}{4} \sum_{v_1} [W(y_1 | 0 , u_1 + v_1) W^{[2]} (y_2|v_1) + W(y_1|1 , u_1 + v_1) W^{[2]}(y_2| v_1)]   \nonumber \\
&= \frac{1}{2} \sum_{v_1}W^{[2]}(y_1 | u_1 + v_1) W^{[2]}(y_2| v_1).   \label{eq:bad_1}
\end{align}

\begin{equation}
(W \boxast W)^{[2]}(y_1, y_2|u_2) =  \frac{(W \boxast W)(y_1, y_2|0 , u_2) + (W \boxast W)(y_1, y_2|1 , u_2)}{2}   \quad \quad \quad \quad \quad \quad \quad \quad \quad \quad \quad \quad  \nonumber \\
\end{equation}
\begin{align}
& = \frac{1}{8} \sum_{v_1 , v_2} [ W(y_1|u_2 ,v_1) W(y_2|u_2 + v_2 ,v_1) + W(y_1|u_2 ,v_1 + 1) W(y_2|u_2 + v_2 ,v_1)] \nonumber \\
&= \frac{1}{8} \sum_{v_1} [W(y_1|u_2 ,v_1) \sum_{v_2} W(y_2|u_2 + v_2 ,v_1) + W(y_1|u_2 ,v_1 + 1) \sum_{v_2} W(y_2|u_2 + v_2 ,v_1)]  \nonumber \\
 &= \frac{1}{4} \sum_{v_1}  [W(y_1|u_2 ,v_1)   +  W(y_1|u_2 ,v_1 + 1)] W^{[2]}(y_2|v_1) \nonumber \\
& = \frac{1}{2} W^{[1]}(y_1|u_2) \sum_{v_1} W^{[2]}(y_2|v_1).   \label{eq:bad_2}
\end{align}


\begin{align}
(W \varoast W)^{[1]}(y_1,y_2, u_1,u_2|v_1) = \frac{(W \varoast W)^{[1]}(y_1,y_2,u_1,u_2|v_1,0) + (W \varoast W)^{[1]}(y_1,y_2|v_1,1)}{2} & \nonumber \\
  = \frac{1}{8} \big[W(y_1|u_2,u_1 + v_1) W(y_2|u_2,v_1 ) + W(y_1|u_2,u_1 + v_1) W(y_2|u_2 + 1,v_1) \big]  & \nonumber \\
=\frac{1}{4} W(y_1|u_2,u_1 + v_1) W^{[2]}(y_2|v_1). \quad \quad \quad \quad  \quad \quad \quad \quad  \quad \quad \quad \quad \quad \quad \quad \quad \quad  &\label{eq:good1} 
\end{align}

\begin{align}
(W \varoast W)^{[2]}(y_1,y_2, u_1,u_2|v_2) = \frac{(W \varoast W)^{[1]}(y_1,y_2,u_1,u_2|0,v_2) + (W \varoast W)^{[1]}(y_1,y_2|1,v_2)}{2}  & \nonumber \\
= \frac{W(y_1|u_2,u_1) W(y_2|u_2+v_2,0)+ W(y_1|u_2,u_1 + 1) W(y_2|u_2+v_2,1)}{8}.  \quad \quad \quad \quad  \quad \quad \quad \quad  &  \label{eq:good2}
\end{align}

\noindent \textbf{Proof of~(\ref{eq::ineq_1})}:
From~(\ref{eq:bad_1}), we have that
\begin{align}
(W \boxast W)^{[1]}(y_1, y_2|0) &= \frac{W^{[2]}(y_1 | 0) W^{[2]}(y_2| 0)  + W^{[2]}(y_1 | 1) W^{[2]}(y_2| 1)}{2} \nonumber \\
(W \boxast W)^{[1]}(y_1, y_2|1) &= \frac{W^{[2]}(y_1 | 0) W^{[2]}(y_2| 1)  + W^{[2]}(y_1 | 1) W^{[2]}(y_2| 0)}{2}  \nonumber 
\end{align}

\noindent Define $\alpha(y_1) = W^{[2]}(y_1 | 0)$, $\beta(y_2) = W^{[2]}(y_2|0)$, $\delta(y_1) = W^{[2]}(y_1|1) $ and $\gamma(y_2) = W^{[2]}(y_2|1)$. Then, the following equalities hold,
\begin{align}
Z(W^{[2]}) &= \sum_{y_1} \sqrt{\alpha(y_1) \delta(y_1)} = \sum_{y_2} \sqrt{\beta(y_2) \gamma(y_2)}\label{eq:notn_bht} \\
\sum_{y_1} \alpha(y_1) &= \sum_{y_1} \delta(y_1) = \sum_{y_2} \beta(y_2) =  \sum_{y_2} \gamma(y_2) = 1 \label{eq:not_prob}
\end{align}
The Bhattacharyya parameter of the partial channel $(W \boxast W)^{[1]}$ is given by
\begin{align*}
Z\big((W \boxast W)^{[1]}\big) &= \sum_{y_1,y_2} \sqrt{(W \boxast W)^{[1]}(y_1, y_2|0)(W \boxast W)^{[1]}(y_1, y_2|1)} \nonumber \\
& = \frac{1}{2} \sum_{y_1, y_2} \sqrt{\alpha(y_1)\beta(y_2) + \delta(y_1)\gamma(y_2)} \sqrt{\alpha(y_1)\gamma(y_2) + \delta(y_1)\beta(y_2)} \nonumber \\
& \leq \frac{1}{2} \sum_{y_1} [\alpha(y_1)+\delta(y_1)] \sum_{y_2} \sqrt{\beta(y_2)\gamma(y_2)} + \frac{1}{2} \sum_{y_1} \sqrt{\alpha(y_1)\delta(y_1)}\sum_{y_2} [\beta(y_2) + \gamma(y_2)] \nonumber \\
& \qquad \qquad \qquad \qquad \qquad \qquad \quad \quad \quad \quad - \sum_{y_1, y_2} \sqrt{\alpha(y_1) \delta(y_1) \beta(y_2) \gamma(y_2)} \nonumber \\
& = 2 Z(W^{[2]}) - Z(W^{[2]})^2,
\end{align*}
where for the third inequality, we have used the following inequality from~\cite{arikan09}
\begin{align}
\sqrt{(\alpha \beta + \delta \gamma) (\alpha\gamma + \delta \beta)} \leq (\sqrt{\alpha\beta} + \sqrt{\gamma \delta} ) (\sqrt{\alpha\gamma} + \sqrt{\delta \beta}) - 2 \sqrt{\alpha\beta\gamma\delta}.
\end{align}
and the fourth equality follows from~(\ref{eq:notn_bht}) and~(\ref{eq:not_prob}). \\
\medskip

\noindent \textbf{Proof of~(\ref{eq::ineq_2})}:
The Bhattacharyya parameter of the partial channel $(W \boxast W)^{[2]}$ is given by
\begin{align*}
Z((W \boxast W)^{[2]}) &= \frac{1}{2} \sum_{y_1, y_2} \sqrt{(W \boxast W)^{[2]}(y_1, y_2|0)} \sqrt{(W \boxast W)^{[2]}(y_1, y_2|1)} \nonumber \\
&= \frac{1}{2} \sum_{y_1} \sqrt{W^{[1]}(y_1|0) W^{[1]}(y_1|1)} \sum_{y_2} \sum_{v_1} W^{[1]}(y_2|v_1) \nonumber \\
&=  Z(W^{[1]}),
\end{align*}
where the second equality follows from~(\ref{eq:bad_2}).

\medskip\noindent \textbf{Proof of~(\ref{eq::ineq_3})}:
The Bhattacharyya parameter of the partial channel $(W \varoast W)^{[1]}$ is given by
\begin{align}
Z\big((W \varoast W)^{[1]}\big) & = \sum_{y_1,y_2, u_1,u_2} \sqrt{(W \varoast W)^{[1]}(y_1,y_2,u_1,u_2|0) (W \varoast W)^{[1]}(y_1,y_2,u_1,u_2|1)} \nonumber \\
&= \frac{1}{4} \sum_{y_1} \sum_{u_1,u_2} \sqrt{W(y_1|u_2,u_1) W(y_1|u_2,u_1+1)} \sum_{y_2} \sqrt{ W^{[2]}(y_2|0)W^{[2]}(y_2|1)} \nonumber \\
&= Z_{1}(W)Z(W^{[2]}),
\end{align}
where the second equality follows from~(\ref{eq:good1}).

\medskip\noindent \textbf{Proof of~(\ref{eq::ineq_4})}:
The Bhattacharyya parameter of the partial channel $(W \varoast W)^{[2]}$ is given by
\begin{align}
Z\big((W \varoast W)^{[2]}\big)  = \sum_{y_1,y_2, u_1, u_2} \sqrt{(W \varoast W)^{[2]}(y_1,y_2,u_1,u_2|0) (W \varoast W)^{[2]}(y_1,y_2,u_1,u_2|1)} \quad \quad \quad & \nonumber \\
\leq \sum_{y_1,y_2,u_2} \sqrt{\sum_{u_1}(W \varoast W)^{[2]}(y_1,y_2,u_1,u_2|0) \sum_{u_1'}(W \varoast W)^{[2] }(y_1,y_2,u_1',u_2|1)} \quad \quad \quad \quad \quad \quad \quad & \nonumber \\
 = \frac{1}{4} \sum_{y_1, y_2, u_2} \sqrt{\frac{1}{2} \left(\sum_{u_1} W(y_1|u_2,u_1)\right) W(y_2|u_2,0)+ \frac{1}{2}\left(\sum_{u_1}W(y_1|u_2,u_1 + 1)\right) W(y_2|u_2,1)} \quad & \nonumber \\
 \qquad  \cdot \sqrt{\frac{1}{2}\left(\sum_{u_1'} W(y_1|u_2,u_1')\right) W(y_2|u_2 + 1,0)+ \frac{1}{2}\left(\sum_{u_1'}W(y_1|u_2,u_1' + 1)\right) W(y_2|u_2+1,1)} & \nonumber \\
 = \frac{1}{2} \sum_{y_1, y_2, u_2} W^{[1]}(y_1|u_2) \sqrt{\frac{W(y_2|u_2,0) + W(y_2|u_2,1)}{2}} \sqrt{\frac{W(y_2|u_2+1,0) + W(y_2|u_2+1,0)}{2}}  & \nonumber \\
 = \frac{1}{2} \sum_{y_2, u_2} \sqrt{W^{[1]}(y_2|u_2) W^{[1]}(y_2|u_2+1)} \quad \quad \quad \quad  \quad \quad \quad \quad  \quad \quad \quad \quad  \quad \quad \quad \quad \quad \quad \quad \quad  \quad \quad &\nonumber \\
 = Z(W^{[1]}), \quad \quad \quad \quad  \quad \quad \quad \quad  \quad \quad \quad \quad  \quad \quad \quad \quad  \quad \quad \quad \quad  \quad \quad \quad \quad \quad \quad \quad \quad  \quad \quad \quad \quad \quad \quad & \nonumber
\end{align}
where for the second inequality, consider vectors {\small $\vec{A}(y_1, y_2, u_2) = ( \sqrt{(W \varoast W)^{[2]}(y_1,y_2,u_1,u_2|0)})_{u_1} $} and $\vec{B}(y_1, y_2, u_2) = ( \sqrt{(W \varoast W)^{[2]}(y_1,y_2,u_1,u_2|1)})_{u_1} $. Then, it follows from the Cauchy -Schwartz inequality, $|  \vec{A}(y_1, y_2, u_2) \cdot \vec{B}(y_1, y_2, u_2)| \leq |\vec{A}(y_1, y_2, u_2)| |\vec{B}(y_1, y_2, u_2)|$. The third equality follows from~(\ref{eq:good2}).

 \printbibliography

\end{document}